\documentclass[11pt,letterpaper]{article}
\usepackage[margin=2.4cm]{geometry}
\usepackage[T1]{fontenc}
\usepackage[utf8]{inputenc}
\usepackage[english]{babel}
\usepackage{tikz,amsmath,amssymb,amsfonts,latexsym,graphicx,amsthm,algorithmicx}
\usepackage{subcaption}
\usepackage{bbm}
\usepackage{comment}
\usepackage{thmtools}
\usepackage{thm-restate}
\declaretheorem{theorem}

\usepackage{algorithm}
\usepackage[noend]{algpseudocode}
\usepackage{hyperref}
\usepackage[capitalise]{cleveref}

\theoremstyle{plain}
\newtheorem{lemma}[theorem]{Lemma}

\newtheorem{fact}[theorem]{Fact}
\newtheorem{corollary}[theorem]{Corollary}
\newtheorem{remark}[theorem]{Remark}
\theoremstyle{definition}

\newtheorem{definition}[theorem]{Definition}

\newcommand{\OPT}{\mathrm{OPT}}
\newcommand{\TSP}{\mathrm{TSP}}
\newcommand{\SOL}{\mathrm{SOL}}
\newcommand{\E}{\mathbb{E}}
\newcommand{\opt}{\mathrm{opt}}
\newcommand{\spine}{\mathrm{spine}}

\newcommand{\dist}{\mathrm{dist}}

\newcommand{\cost}{\mathrm{cost}}

\newcommand{\DP}{\mathrm{DP}}
\newcommand{\eps}{\epsilon}
\newcommand{\biggamma}{\Gamma}

\title{A PTAS for Capacitated Vehicle Routing on Trees\footnote{This is the full version of the extended abstract that was accepted at the 49th EATCS International Colloquium on Automata, Languages, and Programming (ICALP) 2022.}
\vspace{5mm}}
\date{}
\author{Claire Mathieu\footnote{CNRS Paris, France, e-mail: \texttt{claire.mathieu@irif.fr}.}  \and Hang Zhou\footnote{Ecole Polytechnique, Institut Polytechnique de Paris, France, e-mail: \texttt{hzhou@lix.polytechnique.fr}.}}

\begin{document}

\maketitle

\begin{abstract}
    We give a polynomial time approximation scheme (PTAS) for the unit demand capacitated vehicle routing problem (CVRP) on trees, for the entire range of the tour capacity. The result extends to the splittable CVRP.
\end{abstract}

%\thispagestyle{empty}
%\newpage{}
%\setcounter{page}{1}

\section{Introduction}

Given an edge-weighted graph with a vertex called \emph{depot}, a subset of vertices with demands, called \emph{terminals}, and an integer \emph{tour capacity} $k$, the \emph{capacitated vehicle routing problem (CVRP)} asks for a minimum length collection of tours starting and ending at the depot such that those tours together cover all the demand and the total demand covered by each tour is at most $k$.
In the \emph{unit demand} version, each terminal has unit demand, which is covered by a single tour;\footnote{Thus we may identify the demand coverd with the number of terminals covered.}
in the \emph{splittable} version, each terminal has a positive integer demand and the demand at a terminal may be covered by multiple tours.

The CVRP was introduced by Dantzig and Ramser in 1959~\cite{dantzig1959truck} and is arguably one of the most important problems in Operations Research.
Books have been dedicated to vehicle routing problems, e.g.,\ \cite{toth2002vehicle,golden2008vehicle,crainic2012fleet,anbuudayasankar2016models}.
Yet, these problems remain challenging, both from a practical and a theoretical perspective.

Here we focus on the special case when the underlying metric is a tree. That case has been the object of much research.
The splittable tree CVRP was proved NP-hard in 1991~\cite{labbe1991capacitated}, so researchers turned to approximation algorithms.
Hamaguchi and Katoh~\cite{hamaguchi1998capacitated} gave a simple lower bound: every edge must be traversed by enough tours to cover all terminals whose shortest paths to the depot contain that edge.
Based on this lower bound, they designed a 1.5-approximation  in polynomial time~\cite{hamaguchi1998capacitated}.
The approximation ratio  was improved to $(\sqrt{41}-1)/4$ by Asano, Katoh, and Kawashima~\cite{asano2001new} and further to $4/3$ by Becker~\cite{becker2018tight}, both results again based on the lower bound from~\cite{hamaguchi1998capacitated}.
On the other hand, it was shown in~\cite{asano2001new} that using this lower bound one cannot achieve an approximation ratio better than $4/3$.
More recently, researchers tried to go beyond a constant factor so as to get a $(1+\epsilon)$-approximation, at the cost of relaxing some of the constraints.
When the tour capacity is allowed to be violated by an $\eps$ fraction, there is a bicriteria PTAS for the unit demand tree CVRP due to Becker and Paul~\cite{becker2019framework}.
When the running time is allowed to be quasi-polynomial,
Jayaprakash and Salavatipour~\cite{jayaprakash2021approximation} very recently gave a \emph{quasi-polynomial time approximation scheme (QPTAS)} for the unit demand and the splittable versions of the tree CVRP.
In this paper, we close this line of research by obtaining a $(1+\epsilon)$-approximation without relaxing any of the constraints -- in other words, a \emph{polynomial-time approximation scheme (PTAS)}.

\begin{theorem}
\label{thm:main}
There is an approximation scheme for the unit demand capacitated vehicle routing problem (CVRP) on trees with polynomial running time.
\end{theorem}

\begin{corollary}
\label{cor:main}
There is an approximation scheme for the splittable capacitated vehicle routing problem (CVRP) on trees with running time polynomial in the number of vertices $n$ and the tour capacity $k$.
\end{corollary}

To the best of our knowledge, this is the first PTAS for the CVRP in a non-trivial metric and for the entire range of the tour capacity.
Previously, PTASs for small capacity as well as QPTASs were given for the CVRP in several metrics, see \cref{sec:related-work}.

\subsection{Related Work}
Our algorithms build on \cite{jayaprakash2021approximation} and \cite{becker2019framework} but with the addition of
significant new ideas, as we now explain.

\subsubsection{Comparison with the QPTAS in \cite{jayaprakash2021approximation}}
Jayaprakash and Salavatipour noted in~\cite{jayaprakash2021approximation} that
\begin{center}``\emph{it is not clear if it (the QPTAS) can be turned into a
PTAS without significant new ideas.}''
\end{center}

The running time in~\cite{jayaprakash2021approximation} is $n^{O_\eps(\log^4 n)}$.
Where do those four $\log n$ factors in the exponent come from?
At a high level, the QPTAS in~\cite{jayaprakash2021approximation} consists of three parts:  (1) reducing the height of the tree; (2) designing a bicriteria QPTAS;  (3) going from the bicriteria QPTAS to a true QPTAS.
Our approach builds on \cite{jayaprakash2021approximation} but differs from it in each of the three parts, so that in the end we get rid of all of four $\log n$ factors, hence a PTAS.

(1) Jayaprakash and Salavatipour~\cite{jayaprakash2021approximation} reduce the input tree height from $O(n)$ to $O_\eps(\log^2 n)$; whereas instead of the input tree, we consider a \emph{tree of components} (\cref{fact:decomposition}) and reduce its height to $O_\eps(1)$, see \cref{fig:transform-global}.
Pleasingly, the height reduction (\cref{sec:simplifyingtree}) is much simpler than in \cite{jayaprakash2021approximation}.
The analysis differs from \cite{jayaprakash2021approximation} and uses the structure of a near-optimal solution established in Section~\ref{sec:components} and the \emph{bounded distance} property (\cref{def:bounded-distance} and \cref{thm:reduction-bounded-dist}).

(2) In the \emph{adaptive rounding} used in \cite{jayaprakash2021approximation}, they consider the entire range $[1,k]$ of the demands of subtours and partition the subtour demands into \emph{buckets}, resulting in $\Omega_\eps(\log k)$ different subtour demands after rounding.
In our approach, we define \emph{large} and \emph{small} subtours inside components, depending on whether their demands are $\Omega_\eps(k)$ (\cref{def:large-small}).
Then we transform the solution structure to eliminate small subtours (\cref{sec:components}), hence only $O_\eps(1)$ different subtour demands after rounding.
This elimination requires a delicate handling of small subtours.
Thanks to the additional structure, our analysis of the adaptive rounding is simpler than in \cite{jayaprakash2021approximation}, and in particular, we do not need the concept of buckets.

(3) Jayaprakash and Salavatipour show that the \emph{orphan tokens}, which are removed from the tours exceeding capacity, can probably be covered by duplicating a small random set of tours in the optimal solution.
Their approach requires remembering the demands of $\Omega(\log n)$ subtours passing through each edge.\footnote{See the proofs of Lemma~2 and Lemma~3 in the full version of \cite{jayaprakash2021approximation} at \url{https://arxiv.org/abs/2106.15034}.}
To avoid this $\Omega(\log n)$ factor, our approach to cover the orphan tokens (\cref{sec:construction-opt-1}) is different, see \cref{fig:ITP}.
The analysis of our approach (\cref{sec:opt1-feasible,sec:analysis-opt1}) contains several novelties of this paper.

\begin{figure}[t]
    \centering
    \includegraphics[scale=0.25]{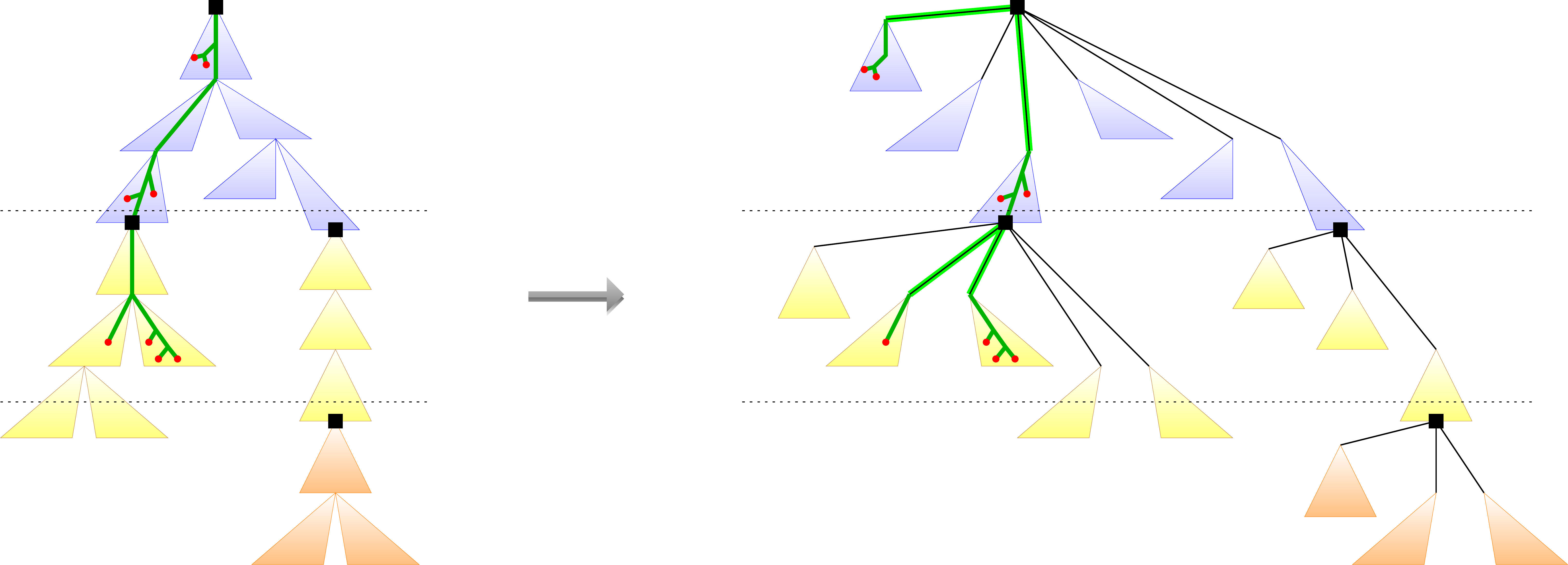}
    \caption{Height reduction for a tree of components.
    The left figure represents the initial tree of components, where each triangle represents a component.
    We partition the components into classes (indicated by blue, yellow, and orange), according to the distances from the roots of the components to the root of the tree, and we reduce the height within each class to 1 (right figure), see~\cref{sec:simplifyingtree}.
    The thick green path in the left figure represents a tour in an optimal solution.
    The red circular nodes are the terminals visited by that tour.
    The corresponding tour in the new tree (right figure)
    spans the same set of terminals.
    %becomes more expensive.
%    We bound the extra cost in \cref{lem:transform-t}.
    }
    \label{fig:transform-global}
\end{figure}

\begin{figure}[t]
    \centering
    \begin{subfigure}[b]{0.55\textwidth}
    \centering
    \includegraphics[scale=0.6]{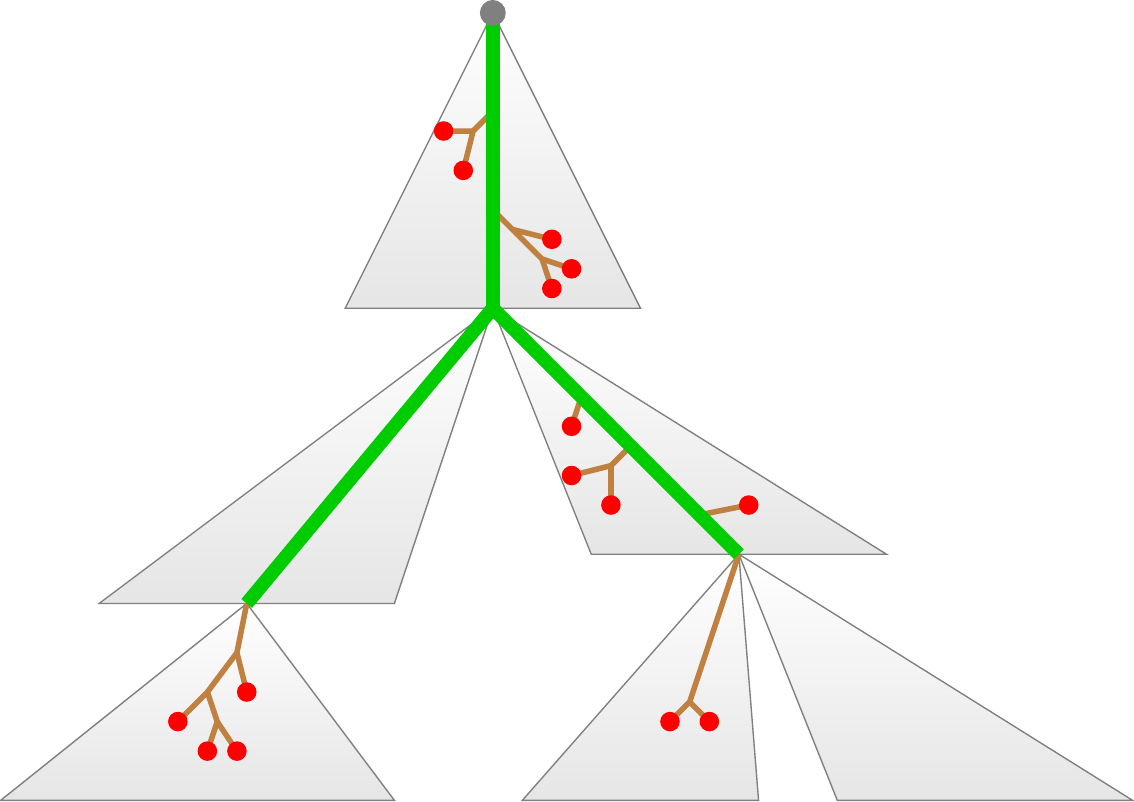}
    \caption{}
    \label{fig:ITP1}
    \end{subfigure}
    \begin{subfigure}[b]{0.4\textwidth}
    \centering
    \includegraphics[scale=0.6]{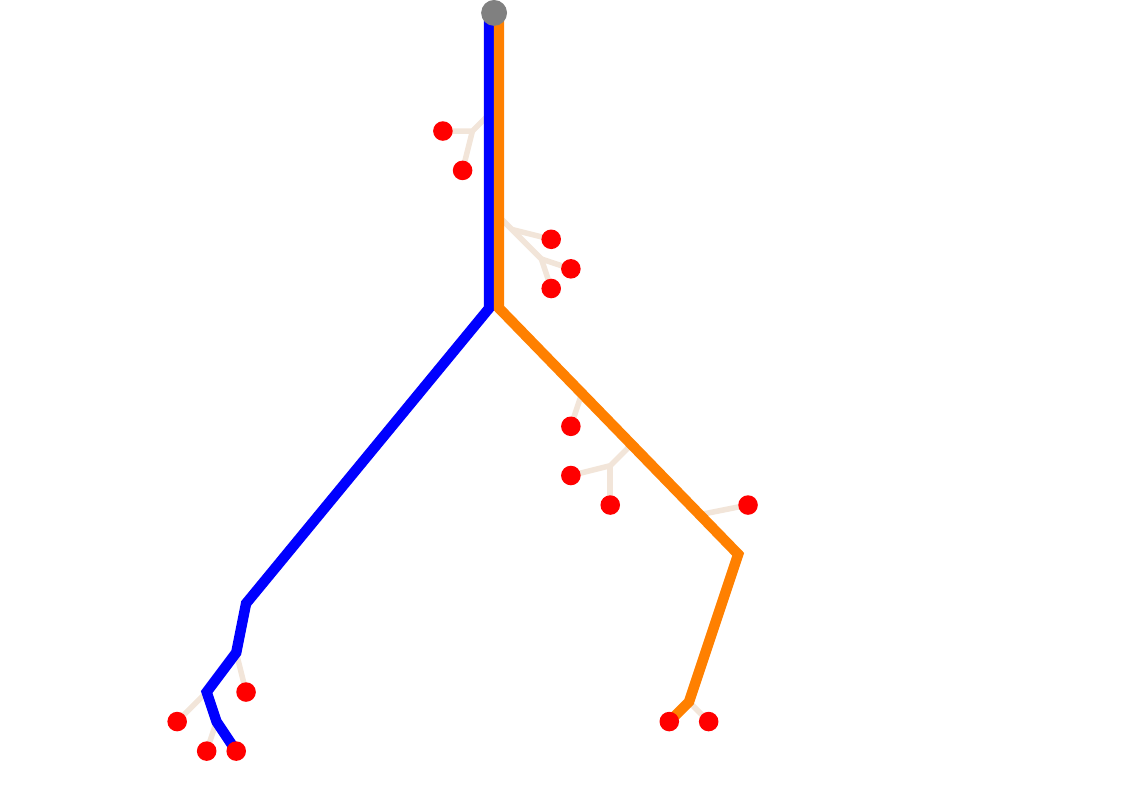}
    \caption{}
    \label{fig:ITP2}
    \end{subfigure}
    \caption{Covering the \emph{orphan tokens} (in red). In \cref{fig:ITP1}, the orphan tokens are contained in the small pieces (in brown) that are removed from the tours exceeding capacity.
    We add the thick paths (in green) to connect all of the small pieces to the root of the tree.
    The cost of the thick paths is an $O(\eps)$ fraction of the optimal cost (\cref{lem:sum-dist-rc,cor:cost-Q}), thanks to the bounded distance property.
    The induced traveling salesman tour is a double cover of the tree spanning the orphan tokens.
    Next, we apply the \emph{iterated tour partitioning (ITP)} algorithm on that tour.
    In \cref{fig:ITP2}, the two paths (in blue and in orange) represent the connections to the depot added by the ITP algorithm.
    Their cost is again an $O(\eps)$ fraction of the optimal cost (\cref{lemma:analyzeITP}), thanks to the bound on the number of orphan tokens and the bounded distance property. %See \cref{sec:analysis-opt1}.
    }
    \label{fig:ITP}
\end{figure}

\subsubsection{Comparison with the Bicriteria PTAS in \cite{becker2019framework}}
Why is the algorithm in \cite{becker2019framework} a bicriteria PTAS, but not a PTAS?

Becker and Paul~\cite{becker2019framework} start by decomposing the tree into \emph{clusters}.
(1) They require that each \emph{leaf} cluster is visited by a single tour.
When the violation of the tour capacity is not allowed, this requirement does not preserve a $(1+\eps)$-approximate solution, see \cref{fig:BP19}.
(2) They also require that each \emph{small} internal cluster is visited by a single tour.
To that end, they modify the optimal solution by reassigning all terminals of a small cluster to some existing tour at the cost of possibly violating the tour capacity.
Such modifications do not seem achievable in the design of a PTAS.

In this paper, we start by defining \emph{components} (\cref{fact:decomposition}), inspired by  \emph{clusters} in \cite{becker2019framework}.
Unlike \cite{becker2019framework}, we allow terminals in any component to be visited by multiple tours.
However, allowing many subtours inside a component could result in an exponential running time for a dynamic program.
To prevent that, we modify the solution structure inside a component so that the number of subtours becomes bounded (\cref{thm:opt1}).
Instead of considering all subtours simultaneously as in \cite{becker2019framework}, we distinguish \emph{small} subtours from \emph{large} subtours (\cref{def:large-small}).
Inspired by \cite{becker2019framework}, we combine small subtours and reallocate them to existing tours such that the violation of the tour capacity is an $O(\eps)$ fraction, see Steps~1 to 3 of the construction in \cref{sec:construction-opt-1}.
Next, we use the \emph{iterated tour partitioning (ITP)} and its postprocessing to reduce the demand of the tours exceeding capacity (\cref{fig:ITP}), which is a novelty in this paper, see Steps~4 to 6 of the construction in \cref{sec:construction-opt-1}.
The ITP algorithm and its postprocessing are analyzed in \cref{sec:opt1-feasible,sec:analysis-opt1}.
In particular, we bound the cost due to the ITP algorithm thanks to the \emph{bounded distance} property (\cref{thm:reduction-bounded-dist}) and to the parameters in our component decomposition that are different from those in~\cite{becker2019framework}, see \cref{remark:components}.
Besides the above novelties in our approach, the height reduction (\cref{fig:transform-global}, see also \cref{sec:simplifyingtree}), the adaptive rounding (\cref{sec:analysisrounding}), the reduction to bounded distances (\cref{sec:preprocessing}), as well as part of the dynamic program (\cref{sec:DP}) are new compared with \cite{becker2019framework}.
These additional techniques are essential in the design of our PTAS, because of the more complicated solution structure inside components in our approach compared with the solution structure inside clusters in \cite{becker2019framework}.

\subsubsection{Other Related Work}
\label{sec:related-work}

\paragraph{Constant-factor approximations in general metric spaces.} The CVRP is a generalization of the \emph{traveling salesman problem (TSP)}.
In general metric spaces, Haimovich and Rinnooy~Kan~\cite{haimovich1985bounds} introduced a simple heuristics, called \emph{iterated tour partitioning (ITP)}.
Altinkemer and Gavish~\cite{altinkemer1990heuristics} showed that the approximation ratio of the ITP algorithm  for the unit demand and the splittable CVRP is at most $1+\left(1-\frac{1}{k}\right)C_{\TSP}$, where $C_{\TSP}\geq 1$ is the approximation ratio of a TSP algorithm.
Bompadre, Dror, and Orlin~\cite{bompadre2006improved} improved this bound to $1+\left(1-\frac{1}{k}\right)C_{\TSP}-\Omega\left(\frac{1}{k^3}\right)$.
The ratio for the unit demand and the splittable CVRP on general metric spaces was recently improved by Blauth, Traub, and Vygen~\cite{blauth2021improving} to $1+C_{\TSP}-\eps$, for some small constant $\eps>0$.

\paragraph{QPTASs.}
Das and Mathieu~\cite{das2015quasipolynomial} designed a QPTAS for the  CVRP in the Euclidean space;
Jayaprakash and Salavatipour~\cite{jayaprakash2021approximation} designed a QPTAS for the CVRP in trees and extended that algorithm to QPTASs in graphs of bounded treewidth, bounded doubling or highway dimension.
When the tour capacity is fixed, Becker, Klein, and Saulpic~\cite{becker2017quasi} gave a QPTAS for planar graphs and bounded-genus graphs.
%, and Cohen-Addad~et~al.~\cite{cohen2020light} gave a QPTAS for minor-free graphs.

\paragraph{PTASs for small capacity.}
In the Euclidean space, there have been PTAS algorithms for the CVRP with small capacity $k$:
 work by Haimovich and Rinnooy Kan~\cite{haimovich1985bounds}, when $k$ is constant; by
 Asano et al.~\cite{asano1997covering}  extending techniques in~\cite{haimovich1985bounds}, for $k=O(\log n/\log\log n)$; and by
Adamaszek, Czumaj, and Lingas~\cite{adamaszek2010ptas}, when  $k\leq 2^{\log^{f(\eps)}(n)}$.
For higher dimensional Euclidean metrics, Khachay and Dubinin~\cite{khachay2016ptas} gave a PTAS for fixed dimension $\ell$ and $k=O(\log^{\frac{1}{\ell}}(n))$.
Again when the capacity is bounded, Becker, Klein and Schild~\cite{becker2019ptas} gave a PTAS for planar graphs; Becker, Klein, and Saulpic~\cite{becker2018polynomial} gave a PTAS for graphs of bounded highway dimension; and Cohen-Addad~et~al.~\cite{cohen2020light} gave PTASs for bounded genus graphs and bounded treewidth graphs.

\paragraph{Unsplittable CVRP.}
In the \emph{unsplittable} version of the CVRP, every terminal has a positive integer demand, and the entire demand at a terminal should be served by a single tour.
On general metric spaces, the best-to-date approximation ratio for the unsplittable CVRP is roughly $3.194$ due to the recent work of Friggstad et al.~\cite{friggstad2021improved}.
For tree metrics, the unsplittable CVRP is APX-hard: indeed, it is NP-hard to approximate the unsplittable tree CVRP to better than a 1.5 factor~\cite{golden1981capacitated} using a reduction from the bin packing problem.
Labbé, Laporte and Mercure~\cite{labbe1991capacitated} gave a 2-approximation for the unsplittable tree CVRP.
The approximation ratio for the unsplittable tree CVRP was improved to $(1.5+\eps)$ very recently by Mathieu and Zhou~\cite{MZ22}, building upon several techniques in the current paper.

%Hence our focus on the unit demand and the splittable versions of the CVRP in the design of a PTAS.

\subsection{Overview of Our Techniques}
\label{sec:overview}
The main part of our work focuses on the unit demand tree CVRP, and we extend our results to the splittable tree CVRP in the end of this work.
%Let $T$ denote the tree.
\begin{definition}[bounded distances]
\label{def:bounded-distance}
Let $D_{\min}$ (resp.\ $D_{\max}$) denote the minimum (resp.\ maximum) distance between the depot and any terminal in the tree.
We say that an instance has \emph{bounded distances} if
$D_{\max}<\left(\frac{1}{\eps}\right)^{\frac{1}{\eps}-1}\cdot D_{\min}.$
\end{definition}

\cref{thm:main} follows directly from \cref{thm:PTAS-bounded-dist,thm:reduction-bounded-dist}.

\begin{theorem}
\label{thm:PTAS-bounded-dist}
There is a polynomial time $(1+4\eps)$-approximation algorithm for the unit demand CVRP on the tree $T$ with bounded distances.
\end{theorem}

%The next theorem holds for the unit demand version, the splittable version, and the unsplittable version of the tree CVRP.
\begin{theorem}
\label{thm:reduction-bounded-dist}
For any $\rho\geq 1$, if there is a polynomial time $\rho$-approximation algorithm for the unit demand (resp.\ splittable, or unsplittable)
CVRP on trees with \emph{bounded distances}, then there is a polynomial time $(1+5\eps) \rho$-approximation algorithm for the unit demand (resp.\ splittable, or unsplittable) CVRP on trees with \emph{general distances}.
\end{theorem}
\cref{thm:reduction-bounded-dist} may be of independent interest for the splittable and the unsplittable versions of the tree CVRP.

\paragraph{Outline of the PTAS for unit demand instances with bounded distances (\cref{thm:PTAS-bounded-dist}).}
In~\cref{sec:components,sec:simplifyingtree,sec:analysisrounding}, we show that there exists a near-optimal solution with a simple structure, and in \cref{sec:DP}, we use a dynamic program to compute the best solution with that structure.

In \cref{sec:components}, we consider the \emph{components} of the tree $T$ (\cref{fact:decomposition}) and we show that there exists a near-optimal solution such that terminals within each component are visited by a constant $O_\eps(1)$ number of tours and that each of those tours visits $\Omega_\eps(k)$ terminals in that component (\cref{thm:opt1}).
The proof of \cref{thm:opt1} contains several novelties in our work.
We start by defining \emph{large} and \emph{small} subtours inside a component, depending on the number of terminals visited by the subtours.
To construct a near-optimal solution with that structure, first, we detach small subtours from their initial tours, combine small subtours in the same component, and reallocate the combined subtours to existing tours.
Then we remove subtours from tours exceeding capacity.
To connect the removed subtours to the root of the tree, we include the \emph{spines subtours} (\cref{def:spine}) of all \emph{internal components}, and we obtain a traveling salesman tour.
Next, we apply the \emph{iterated tour partitioning (ITP)} algorithm on that tour, see \cref{fig:ITP}.
Finally, in a postprocessing step, we eliminate the small subtours created due to the ITP algorithm.
The complete construction is in \cref{sec:construction-opt-1}; the feasibility of the construction is in \cref{sec:opt1-feasible}; and the analysis on the constructed solution is in \cref{sec:analysis-opt1}, which in particular uses the bounded distance property.

In \cref{sec:simplifyingtree}, we transform the tree $T$ into a tree $\hat T$ that has $O_\eps (1)$ levels of components (\cref{fig:transform-global}) and satisfies the following property.
\begin{fact}
\label{fact:hat-T}
The tree $\hat T$ defined in \cref{sec:simplifyingtree} can be computed in polynomial time.
The components in the tree $\hat T$ are the same as those in the tree $T$.
Any solution for the unit demand CVRP on the tree $\hat T$ can be transformed in polynomial time into a solution for the unit demand CVRP on the tree $T$ without increasing the cost.
\end{fact}
\noindent Thanks to the structure of the near-optimal solution on $T$ (\cref{sec:components}) and to the bounded distance property, the optimal cost for $\hat T$ is increased by an $O(\eps)$ fraction compared with the optimal cost for $T$ (\cref{thm:opt2}).

\begin{figure}[t]
    \centering
    \includegraphics[scale=0.3]{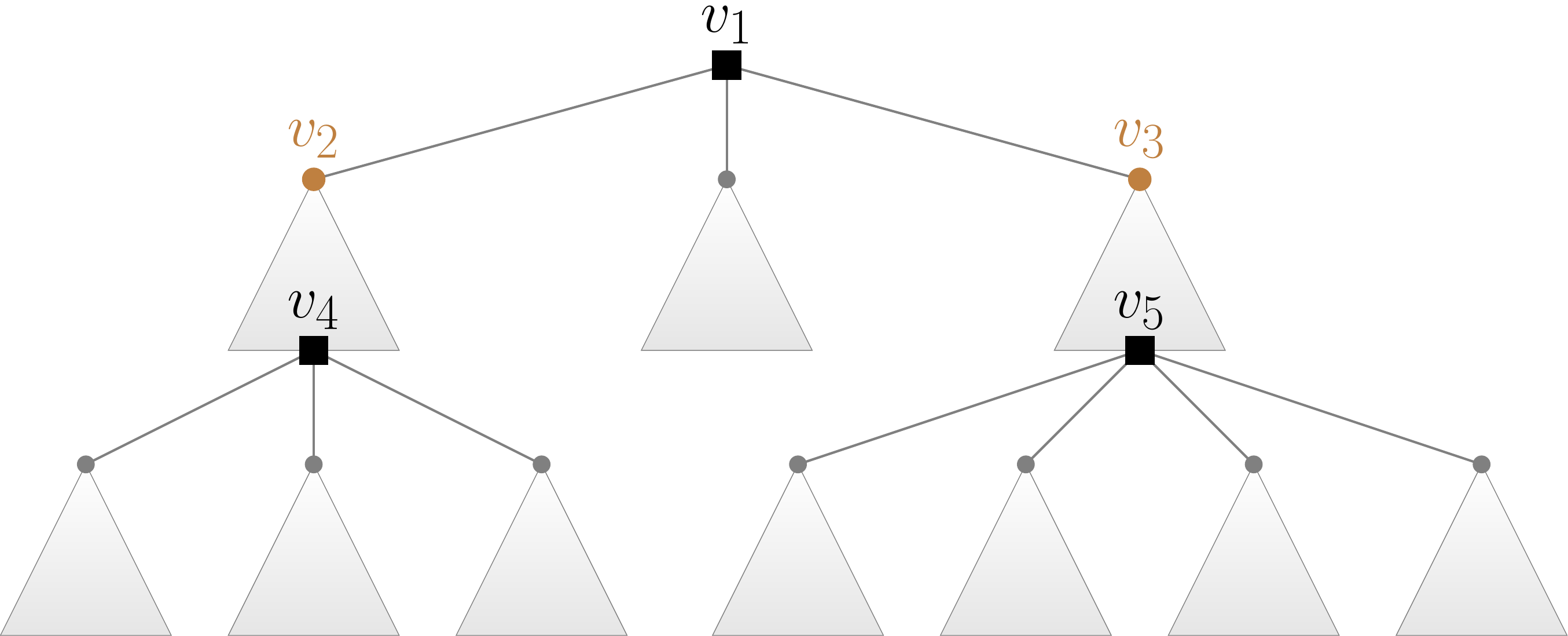}
    \caption{At each vertex of the tree, the dynamic program memorizes the capacities used by the subtours in the subtree and their total cost.
    Terminals within each component are visited by $O_\eps(1)$ tours and each of those tours visits $\Omega_\eps(k)$ terminals in that component.
    Here is an example of the flow of execution in the dynamic program.
    First, independent computation in each component (\cref{alg:local}). Next, computation in the subtrees rooted at vertices $v_4$ and $v_5$ (\cref{alg:subtree-configuration-critical}).
    Then computation for the subtrees rooted at vertices $v_2$ and $v_3$ (\cref{alg:subtree-configuration-root}).
    Finally, computation for the subtree rooted at vertex $v_1$ (\cref{alg:subtree-configuration-critical}).
    The output is the best solution in the subtree
    rooted at $v_1$. Vertices $v_1,v_4$, and $v_5$, whose degrees may be  arbitrarily large, are where adaptive rounding of the subtour demands  is needed to maintain polynomial time (Algorithm~\ref{alg:subtree-configuration-critical}). The cost due to the rounding is small thanks to the way components are defined and the bounded distance property, and does not accumulate excessively because the height is bounded (\cref{sec:simplifyingtree}).
    }
    \label{fig:DP}
\end{figure}

In \cref{sec:analysisrounding}, we apply the \emph{adaptive rounding} on the demands of the subtours in a near-optimal solution on $\hat T$.
Recall that in the design of the QPTAS by Jayaprakash and Salavatipour~\cite{jayaprakash2021approximation}, the main technique is to show the existence of a near-optimal solution in which the demand of a subtour can be rounded to the nearest value from a set of only \emph{poly-logarithmic} threshold values.
In our work, we reduce the number of threshold values to \emph{a constant $O_\eps(1)$} (\cref{thm:opt3}).
To analyze the adaptive rounding,
observe that an extra cost occurs whenever we detach a subtour and complete it into a separate tour by connecting it to the depot.
We bound the extra cost thanks to the structure of a near-optimal solution inside components (\cref{sec:components}), the reduced height of the components (\cref{sec:simplifyingtree}), as well as the bounded distance property.

In \cref{sec:DP}, we design a \emph{polynomial-time dynamic program}  that computes the best solution on the tree $\hat T$ that satisfies the constraints on the solution structure imposed by previous sections.
The algorithm combines a dynamic program inside components (\cref{sec:local}) and two dynamic programs in subtrees (\cref{sec:subtree}), see \cref{fig:DP}.
Thus we obtain the following \cref{thm:DP}.
\begin{theorem}
\label{thm:DP}
Consider the unit demand CVRP on the tree $\hat T$.
There is a dynamic program that computes in polynomial time a solution with cost at most $(1+4\eps)\cdot \opt$, where $\opt$ denotes the optimal cost for the unit demand CVRP on the tree $T$.
\end{theorem}

\cref{thm:PTAS-bounded-dist} follows directly from \cref{thm:DP} and \cref{fact:hat-T}.

\paragraph{Reduction from general distances to bounded distances (\cref{thm:reduction-bounded-dist}).}
In \cref{sec:preprocessing}, we prove \cref{thm:reduction-bounded-dist}.
We use Baker's technique to split tours into pieces such that each piece covers terminals that are within a certain range of distances from the depot. This requires duplicating some parts of the tours so that each piece of the tour is connected to the depot.

\paragraph{Extension to the splittable setting (\cref{cor:main}).}
In \cref{sec:splittable}, we extend the result in \cref{thm:main} to the splittable setting, thus obtaining \cref{cor:main}.

\paragraph{Open questions.}
Previously, Jayaprakash and Salavatipour~\cite{jayaprakash2021approximation} extended their QPTAS on trees to QPTASs on graphs of bounded treewidth and beyond, including Euclidean spaces.
While some of our techniques extend to those settings, others do not seem to  carry over without significant additional ideas, so it is an interesting open question whether the techniques in our paper could be used in the design of PTAS algorithms for other metrics, such as graphs of bounded treewidth, planar graphs, and Euclidean spaces.

\section{Preliminaries}
\label{sec:preliminary}

Let $T$ be a rooted tree $(V,E)$ with root $r\in V$ and edge weights $w(u,v)\geq 0$ for all $(u,v)\in E$.
The root $r$ represents the \emph{depot} of the tours.
Let $n$ denote the number of vertices in $V$.
Let $V'\subseteq V$ denote the set of \emph{terminals}, such that a \emph{token} is placed on each terminal $v\in V'$.
Let $k\in[1,n]$ be an integer \emph{capacity} of the tours.
The \emph{cost} of a tour $t$, denoted by $\cost(t)$, is the overall weight of the edges on that tour.
We say that a tour \emph{visits} a terminal $v\in V'$ if the tour picks up the token at $v$.\footnote{Note that a tour might go through a terminal $v$ without picking up the token at $v$.}
\begin{definition}[unit demand tree CVRP]
An instance of the \emph{unit demand} version of the \emph{capacitated vehicle routing problem (CVRP)} on \emph{trees} consists of
\begin{itemize}
    \item an edge weighted \emph{tree} $T=(V,E)$ with $n=|V|$ and with \emph{root} $r\in V$ representing the \emph{depot},
    \item a set $V'\subseteq V$ of \emph{terminals},
    \item a positive integer \emph{tour capacity} $k$ such that $k\leq n$.
\end{itemize}
A feasible solution is a set of tours such that
\begin{itemize}
    \item each tour starts and ends at $r$,
    \item each tour visits at most $k$ terminals,
    \item each terminal is visited by one tour.
\end{itemize}
The goal is to find a feasible solution such that the total cost of the tours is minimum.

\end{definition}
Let $\OPT$ (resp.\ $\OPT_1$, $\OPT_2$, or $\OPT_3$) denote an optimal (resp.\ near-optimal) solution to the unit demand CVRP, and let $\opt$ (resp.\ $\opt_1$, $\opt_2$, or $\opt_3$) denote the value of that solution.

Without loss of generality, we assume that every vertex in the input tree $T$ has exactly two children, and that the terminals are the same as the leaf vertices of the tree.
Indeed, general instances can be reduced to instances with these properties by inserting edges of weight 0, removing leaf vertices that are not terminals, and slicing out internal vertices of degree two, see, e.g.~\cite{becker2019framework} for details.

For any vertex $v\in V$, a \emph{subtour at the vertex $v$} is a path that starts and ends at $v$ and only visits vertices  in the subtree rooted at $v$.
The \emph{demand} of a subtour is the number of terminals visited by that subtour.
For each vertex $v\in V$, let $\dist(v)$ denote the distance between $v$ and the depot in the tree $T$.
For technical reasons, we allow \emph{dummy} terminals to be included in the solution at internal vertices of the tree.

%\paragraph{Constants in our algorithm.}
Throughout the paper, we define several constants depending on~$\eps$: $\Gamma$ (\cref{fact:decomposition}), $\alpha$ (\cref{thm:opt1}), $H_\eps$ (\cref{lem:tild-D-H-eps}), and $\beta$ (\cref{thm:opt3}).
They satisfy the relation that $H_\eps\gg\Gamma\gg 1\gg \eps\gg \alpha\gg\beta$.

%\paragraph{Iterated Tour Partitioning (ITP) algorithm (\cite{haimovich1985bounds,altinkemer1990heuristics}).}
%In the ITP algorithm, we start by computing a traveling salesman tour (ignoring the capacity constraint) using some other algorithm.
%Then we partition the tour into segments with exactly $k$ terminals each, except possibly the first and the last segments containing less than $k$ terminals.
%Finally, for each segment, we connect its endpoints to the depot so as to make a tour.

\paragraph{Decomposition of the Tree into Components.}
The component decomposition (\cref{fact:decomposition}) is inspired by the cluster decomposition by Becker and Paul~\cite{becker2019framework}.
The proof of \cref{fact:decomposition} is similar to arguments in~\cite{becker2019framework}; we give its proof in \cref{app:decomposition} for completeness.

\begin{figure}
    \centering
    \includegraphics[scale=0.5]{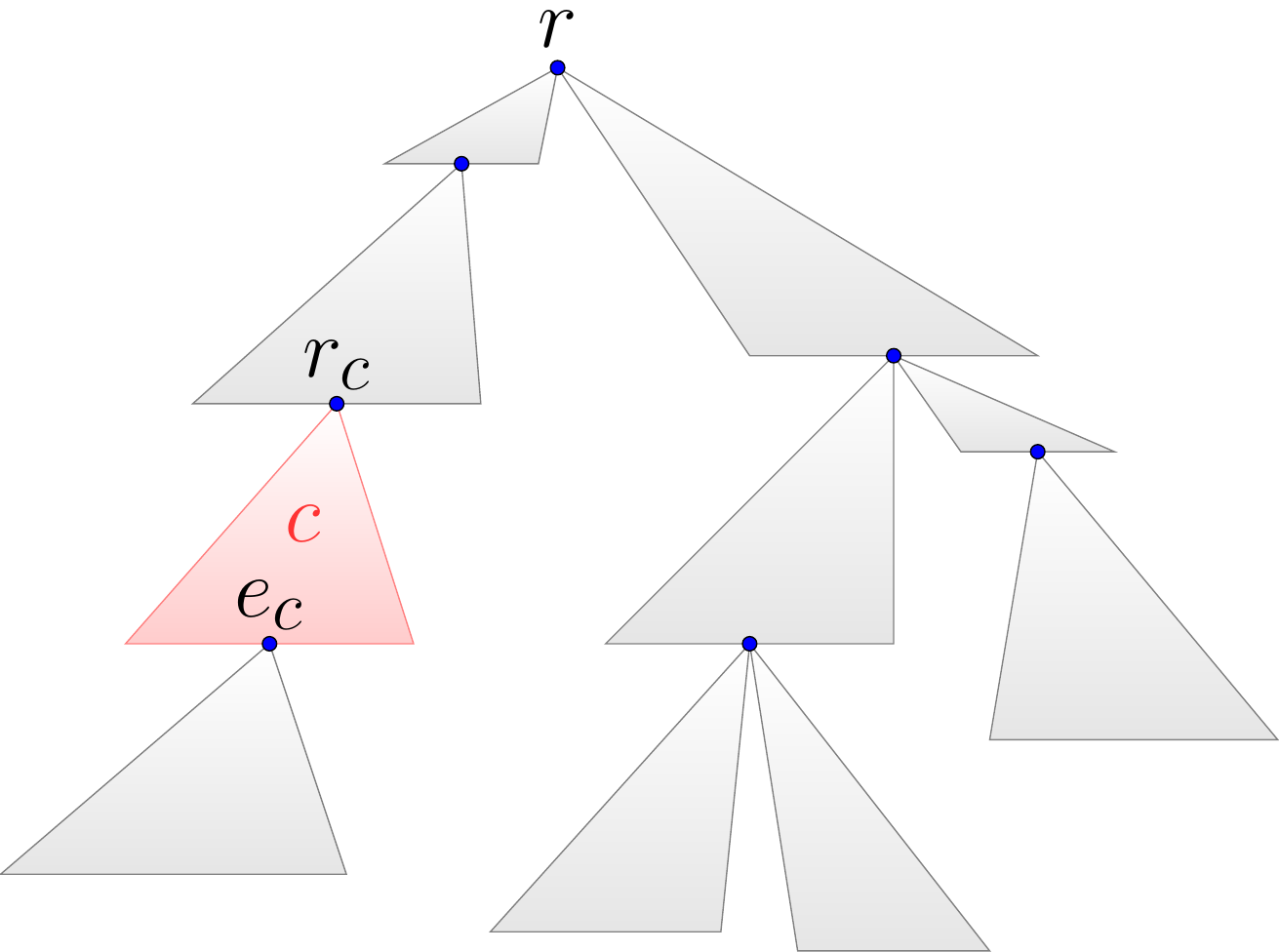}\hspace{2cm}
    \caption{Decomposition into components. In this example, the tree is decomposed into four \emph{leaf} components and six \emph{internal} components.
    An internal component $c$ has a \emph{root} vertex $r_c$ and an \emph{exit} vertex $e_c$.
}
    \label{fig:component}
\end{figure}

\begin{lemma}
\label{fact:decomposition}
Let $\biggamma=\frac{12}{\eps}$.
There is a polynomial time algorithm to compute a partition of the edges of the tree $T$ into a set $\mathcal{C}$ of \emph{components} (see \cref{fig:component}), such that all of the following properties are satisfied:
\begin{itemize}
    \item Every component $c\in \mathcal{C}$ is a connected subgraph of $T$;
    the \emph{root} vertex of the component $c$, denoted by $r_c$, is the vertex in $c$ that is closest to the depot.
    \item We say that a component $c\in \mathcal{C}$ is a \emph{leaf} component if all descendants of $r_c$ in tree $T$ are in $c$, and is
    an \emph{internal} component otherwise. A leaf component $c$ interacts with other components at vertex $r_c$ only.
    An internal component $c$ interacts with other components at two vertices only: at vertex $r_c$, and at another vertex, called the \emph{exit} vertex of the component $c$, and denoted by $e_c$.
    \item Every component $c\in \mathcal{C}$ contains at most $2\biggamma\cdot k$ terminals.
    We say that a component is \emph{big} if it contains at least $\biggamma\cdot k$ terminals.
    Each leaf component is big.
    \item If the number of components in $\mathcal{C}$ is strictly greater than one, then we have: (1) there exists a map from all components to big components, such that the image of a component is among its descendants (including itself), and each big component has at most three pre-images; and (2) the number of components in $\mathcal{C}$ is at most $3/\Gamma$ times the total demand in the tree $T$.
\end{itemize}
\end{lemma}

\begin{remark}
\label{remark:components}
The root and the exit vertices of components are a rough analog of \emph{portals} used in approximation schemes for other problems: they are places where the dynamic program will gather and synthesize information about partial solutions before passing it on.

Compared with \cite{becker2019framework},  the decomposition in \cref{fact:decomposition} uses different parameters:
the number of terminals inside a \emph{leaf component} is $\Theta(k/\eps)$, whereas in \cite{becker2019framework} the number of terminals inside a \emph{leaf cluster} is $\Theta(\eps\cdot k)$; the threshold demand to define \emph{big components} is $\Theta(k/\eps)$, whereas in \cite{becker2019framework} the threshold demand to define \emph{small clusters} is $\Theta(\eps^2\cdot k)$.
\end{remark}

\begin{definition}[subtours in components and subtour types]
Let $c$ be any component.
A \emph{subtour in the component $c$} is a path that starts and ends at the root $r_c$ of the component, and such that every vertex on the path is in $c$.
The \emph{type} of a subtour is ``\emph{passing}'' if $c$ is an internal component and the exit vertex $e_c$ belongs to that subtour; and is ``\emph{ending}'' otherwise.
\end{definition}
\noindent A passing subtour in $c$ is to be combined with a subtour at $e_c$.
In a leaf component, there is no passing subtour.

\begin{definition}[spine subtour]
\label{def:spine}
For an internal component $c$, we define the \emph{spine subtour} in the component $c$, denoted by $\spine_c$, to be the connection (in both directions) between  the root vertex $r_c$ and  the exit vertex $e_c$ of the component $c$, without visiting any terminal.
\end{definition}
\noindent From the definition, a spine subtour in a component is also a passing subtour in that component.

Without loss of generality, we assume that any subtour in a component $c$ either visits at least one terminal in $c$ or is a spine subtour; that any tour traverses each edge of the tree at most twice (once in each direction); and that any tour contains at most one subtour in any component.\footnote{If a tour contains several subtours in a component $c$, we may combine those subtours into a single subtour in $c$ without increasing the cost of the tour.}

\section{Solutions Inside Components}
\label{sec:components}

In this section, we prove \cref{thm:opt1}, which is a main novelty in this paper.

\begin{theorem}
\label{thm:opt1}
Let $\alpha= \eps^{(\frac{1}{\eps}+1)}$.
Consider the unit demand CVRP on the tree $T$ with bounded distances.
There exist dummy terminals and a solution $\OPT_1$ visiting all of the real and the dummy terminals, such that all of the following holds:
\begin{enumerate}
\item For each component $c$, there are at most $\frac{2\Gamma}{\alpha}+1$ tours visiting terminals in $c$;
\item For each component $c$ and each tour $t$ visiting terminals in $c$, the number of the terminals in $c$ visited by $t$ is at least $\alpha\cdot k$;
\item We have $\opt_1\leq (1+\eps)\cdot\opt$.
\end{enumerate}
\end{theorem}

\subsection{Construction of $\OPT_1$}
\label{sec:construction-opt-1}
\begin{definition}[large and small subtours]
\label{def:large-small}
We say that a subtour is \emph{large} if its demand is at least $\alpha\cdot k$, and \emph{small} otherwise.
\end{definition}

The construction of $\OPT_1$, starting from an optimal solution $\OPT$, is in several steps.

\paragraph{Step 1: Detaching small subtours.} Prune each tour of $\OPT$ so that it only visits the terminals that do \emph{not} belong to a small subtour in any component, and is minimal.
In other words, if a tour $t$ in $\OPT$ contains a small ending subtour $t_e$ in a component $c$, then we remove $t_e$ from $t$; and if a tour $t$ in $\OPT$ contains a non-spine small passing subtour $t_p$ in a component $c$, then we remove $t_p$ from $t$, except for the spine subtour of $c$.

Let $A$ denote the set of the resulting tours.
Note that each tour in $A$ is connected.
The removed pieces of a non-spine small passing subtour $t_p$ may be disconnected from one another.

The parts of $\OPT$ that have been pruned consist of the set $\mathcal{E}$, each element being a small ending subtours in a component, and of the set $\mathcal{P}$, each element being a group of pieces in a component obtained from a non-spine small passing subtour by removing the corresponding spine subtour.
The \emph{demand} of a group of pieces in $\mathcal{P}$ is the number of terminals in all of the pieces in that group.

\paragraph{Step 2: Combining small subtours within components.}
    \begin{itemize}
    \item
    For each component $c$, repeatedly concatenate subtours in $c$ from the set $\mathcal{E}$ so that in the end, all resulting subtours in $c$ from the set $\mathcal{E}$ have demands between $\alpha\cdot k$ and $2\alpha\cdot k$ except for at most one subtour.
    \item
    For each component $c$, repeatedly merge groups in $c$ from the set $\mathcal{P}$ so that in the end, all resulting groups in $c$ from the set $\mathcal{P}$ have demands between $\alpha\cdot k$ and $2\alpha\cdot k$ except for at most one group.
    \end{itemize}

    Let $\mathcal{E}'$ and $\mathcal{P}'$ denote the corresponding sets after the modifications for all components $c$.
    Let $B=\mathcal{E}'\cup \mathcal{P}'$.
    An \emph{element} of $B$ is either a \emph{subtour} or a \emph{group of pieces} in a component $c$.
    For each component $c$, all elements of $B$ in the component $c$ have demands between $\alpha\cdot k$ and $2\alpha\cdot k$ except for at most two elements with smaller demands.

\paragraph{Step 3: Reassigning small subtours.} Construct a bipartite graph with vertex sets $A$ and $B$ and with edge set $E$.
Let $a$ be any tour in $A$, and let $a_0$ denote the corresponding tour in $\OPT$.
Let $b$ be any element in $B$.
The set $E$ contains an edge between $a$ and $b$ if and only if the element $b$ contains terminals from $a_0$; the weight of the edge $(a,b)$ is the number of terminals in both $b$ and $a_0$.
By Lemma 1 from \cite{becker2019framework}, there exists an assignment $f:B\rightarrow A$  such that each element $b\in B$ is assigned to a tour $a\in A$ with $(a,b)\in E$ and that, for each tour $a\in A$, the demand of $a$ plus the overall demand of the elements $b\in B$ that are assigned to $a$ is at most the demand of the corresponding tour $a_0$ plus the maximum demand of any element in $B$.

Let $A_1$ denote the set of \emph{pseudo-tours} induced by the assignment $f$.
Each pseudo-tour in $A_1$ is the union of a tour $a\in A$ and the elements $b\in B$ that are assigned to $a$.

\paragraph{Step 4: Correcting tour capacities.}
For each pseudo-tour $a_1$ in $A_1$, if the demand of $a_1$ exceeds $k$, we repeatedly remove an element $b\in B$ from $a_1$, until the demand of $a_1$ is at most $k$.

Let $A_2$ denote the resulting set of pseudo-tours.
Every pseudo-tour in $A_2$ is a connected tour of demand at most $k$ (\cref{lem:opt1-feasible}).
Let $B^*\subseteq B$ denote the set of the removed elements $b\in B$.
The elements in $B^*$ are represented by the small pieces in \cref{fig:ITP1} (\cpageref{fig:ITP1}).

\paragraph{Step 5: Creating additional tours.}
We connect the elements of $B^*$ to the depot by creating additional tours as follows.
    \begin{itemize}
    \item Let $Q$ denote the union of the spine subtours for all internal components.
    $Q$ is represented by the green thick paths in \cref{fig:ITP1}.
    Let $X$ denote a multi-subgraph of $T$ that is the union of the elements in $B^*$ and the edges in $Q$.
    Observe that each element in $B^*$ is connected to the depot through edges in $Q$.
    Thus $X$ is connected, and induces a traveling salesman tour $t_{\TSP}$ visiting all terminals in $B^*$.
    Without loss of generality, we assume that, for any component $c$, if $t_{\TSP}$ visits terminals in $c$, then those terminals belong to a single subpath $p_c$ of $t_{\TSP}$, such that $p_c$ does not visit any terminal from other components.
    \item If the traveling salesman tour $t_{\TSP}$ is within the tour capacity, then let $A_3$ denote the set consisting of a single tour $t_{\TSP}$.
    Otherwise, we apply the \emph{iterated tour partitioning (ITP)} algorithm \cite{haimovich1985bounds} on $t_{\TSP}$:
    we partition $t_{\TSP}$ into segments with exactly $k$ terminals each, except possibly the last segments containing less than $k$ terminals.
    For each segment, we connect its endpoints to the depot so as to make a tour, see \cref{fig:ITP2}.
    Let $A_3$ denote the resulting set of tours.
    \end{itemize}

    Let $A_4=A_2\cup A_3$.

\paragraph{Step 6: Postprocessing.}
For each component $c$, we rearrange the small subtours in $A_4$ as follows.
\begin{itemize}
    \item For each tour $t$ in $A_4$ that contains a small subtour in $c$, letting $t_c$ denote this small subtour, if $t_c$ is an ending subtour, we remove $t_c$ from $t$; if $t_c$ is a passing subtour, we remove $t_c$ from $t$, except for the spine subtour in $c$.
    The total demand of all of the removed small subtours in $c$ is at most $k$ (\cref{lemma:small-subtours}).
    \item We create an additional tour $t^*_c$ from the depot that connects all of the removed small subtours in $c$.
    If the demand of $t^*_c$ is less than $\alpha\cdot k$, then we include dummy terminals at $r_c$ into the tour $t^*_c$ so that its demand is exactly $k$.
\end{itemize}
Let $A_5$ denote the resulting set of tours after removing small subtours from $A_4$.
Let $A_6$ denote the set of newly created tours $\{t^*_c\}_{c\in \mathcal{C}}$.

Finally, let $\OPT_1=A_5\cup A_6$.

In \cref{sec:opt1-feasible}, we show that $\OPT_1$ is a feasible solution to the unit demand tree CVRP;
in \cref{sec:analysis-opt1}, we prove the three properties of $\OPT_1$ in the claim of \cref{thm:opt1}.

\subsection{Feasibility of the Construction}
\label{sec:opt1-feasible}
\begin{lemma}
\label{lem:opt1-feasible}
Every pseudo-tour in $A_2$ is a connected tour of demand at most $k$.
\end{lemma}

\begin{proof}
Let $a_2$ be any pseudo-tour in $A_2$.
By the construction in Step~4, the demand of $a_2$ is at most $k$.
It suffices to show that $a_2$ is a connected tour.

Observe that $a_2$ is the union of a tour $a\in A$ and some elements $b_1,\dots,b_q$ from $B$, for $q\geq 0$.
From the construction, any tour $a\in A$ is connected.
Consider any $b_i$ for $i\in [1,q]$.
Observe that $f(b_i)=a$, so the edge $(a,b_i)$ belongs to the bipartite graph $(A,B)$.
If $b_i$ is a subtour at $r_c$ for some component $c$, then $r_c$ must belong to the tour $a$; and if $b_i$ is a group of pieces in some component $c$, then the spine subtour of $c$ must belong to the tour $a$.
So the union of $a$ and $b_i$ is connected.
Thus $a_2=a\cup b_1\cup\dots\cup b_q$ is a connected tour.
\end{proof}

\begin{lemma}\label{lemma:small-subtours}
In any component $c$, the total demand of all of the removed small subtours in $c$ at the beginning of Step~6 is at most $k$.
\end{lemma}

\begin{proof}
Let $c$ be any component.
The key is to show that the number of non-spine small subtours in $c$ that are contained in tours in $A_4$ is at most 4.
Since $A_4=A_2\cup A_3$, we analyze the number of non-spine small subtours in $c$ that are contained in tours in $A_2$ and in $A_3$, respectively.

The tours in $A_2$ contain at most two non-spine small subtours in $c$, since at most two elements of $B$ in component $c$ have demands less than $\alpha\cdot k$.

We claim that the tours in $A_3$ contain at most two non-spine small subtours in $c$.
If $A_3$ consists of a single tour $t_{\TSP}$, the claim follows trivially since any tour contains at most one subtour in $c$ from our assumption.
Next, we consider the case when $A_3$ is generated by the ITP algorithm.
From our assumption, if $t_{\TSP}$ visits terminals in $c$, then those terminals belong to a single subpath $p_c$ of $t_{\TSP}$, such that $p_c$ does not visit any terminal from other components.
By applying the ITP algorithm  on $t_{\TSP}$, we obtain a collection of segments.
All segments that intersect $p_c$ visit exactly $k$ terminals in $c$, except for possibly the first and the last of those segments visiting less than $k$ terminals in $c$.
Hence at most two non-spine small subtours in $c$ among the tours in $A_3$.

Therefore, the number of non-spine small subtours in $c$ in the solution $A_4$ is at most 4.
Since each small subtour has demand at most $\alpha\cdot k$, the total demand of the removed small subtours is at most $4\cdot \alpha\cdot k<k$.
\end{proof}

\subsection{Analysis of $\OPT_1$}
\label{sec:analysis-opt1}
Let $c\in \mathcal{C}$ be any component.
From \cref{fact:decomposition}, $c$ contains at most $2\Gamma\cdot k$ real terminals.
Each tour in $A_5$ visiting terminals in $c$ visits at least $\alpha\cdot k$ real terminals in $c$, so there are at most  $\frac{2\Gamma}{\alpha}$ tours in $A_5$ visiting terminals in $c$.
There is a single tour in $A_6$, the tour $t_c^*$, that visits terminals in $c$.
Hence the first property of the claim.
From the construction of $t_c^*$, the second property of the claim follows.

It remains to analyze the cost of $\OPT_1$.
Compared with $\OPT$, the extra cost in $\OPT_1$ is due to Step~5 and Step~6 of the construction.
This extra cost consists of the cost of the edges in $Q$ (Step~5), the cost in the ITP algorithm to connect the endpoints of all segments to the depot (Step~5), and the cost of the postprocessing (Step~6), which we bound in \cref{cor:cost-Q,lemma:analyzeITP,lem:connect-small}, respectively.
Both \cref{cor:cost-Q} and \cref{lem:connect-small} are based on \cref{lem:sum-dist-rc}.

%\begin{lemma}\label{lemma:analyzeMST}
%The cost of the subgraph $Q$ at most $\eps\cdot \opt$.
%\end{lemma}

%\begin{proof}
%Consider any edge $e$ in $Q$.
%Let $n_e$ denote the number of tours in $\OPT$ that contain $e$.
%To lower bound $n_e$, observe that there exists some leaf component $c$ such that $e$ belongs to the path between the root of the tree and the root of $c$.
%Since $c$ is a leaf component, $c$ contains at least $\biggamma\cdot k$ terminals by \cref{fact:decomposition}, so $n_e\geq \Gamma$.
%Summing over all edges $e$ in $Q$, we have
%\[\opt\geq \sum_{e\in Q} 2\cdot w(e)\cdot n_e\geq 2\Gamma\cdot \sum_{e\in Q} w(e)=2\Gamma\cdot \cost(Q).\]
%Since $2\biggamma> 1/\eps$, the claim follows.
%\end{proof}

\begin{lemma}
\label{lem:sum-dist-rc}
We have $\displaystyle\sum_{\text{component $c$}} \dist(r_c)\leq \frac{\eps}{8}\cdot\opt.$
\end{lemma}

\begin{proof}
For any edge $e$ in $T$, let $u$ and $v$ denote the two endpoints of $e$ such that $u$ is the parent of $v$.
Let $T_e$ denote the subtree of $T$ rooted at $v$.
Let $n_e$ denote the number of terminals in $T_e$.
From the lower bound in \cite{hamaguchi1998capacitated}, we have
\[\opt\geq \sum_{e\in T}2\cdot w(e)\cdot \frac{n_e}{k}.\]
Since each big component contains at least $\Gamma\cdot k$ terminals, we have
\[n_e\geq \sum_{\text{big component }c\subseteq T_e} \Gamma\cdot k.\]
Thus
\begin{align*}
\opt&\geq \sum_{e\in T}2\cdot w(e)\cdot \sum_{\text{big component }c\subseteq T_e} \Gamma\\
&= \sum_{\text{big component }c} \Gamma \cdot\sum_{e\in T \text{ such that }c\subseteq T_e}2\cdot w(e)\\
&=\sum_{\text{big component }c} \Gamma \cdot 2\cdot\dist(r_c).
\end{align*}
From \cref{fact:decomposition}, there exists a map from all components to big components such that the image of a component is among its descendants (including itself) and each big component has at most three pre-images.
%From Theorem~15 in~\cite{MZ21}, in each component, the number of small subtours is at most 4.
Thus
\[3\cdot \sum_{\text{big component $c$}}  \dist(r_c)\geq \sum_{\text{component $c$}} \dist(r_c).\]
Therefore,
\[\opt\geq \frac{2\cdot\Gamma}{3}\cdot\sum_{\text{component $c$}} \dist(r_c).\]
The claim follows since $\Gamma=\frac{12}{\eps}$.
\end{proof}

\begin{corollary}
\label{cor:cost-Q}
We have $\cost(Q)\leq \frac{\eps}{4}\cdot\opt$.
\end{corollary}

\begin{proof}
Observe that every edge in $Q$ belongs to the connection (in both directions) between the depot and the root of some component $c$.
By \cref{lem:sum-dist-rc}, we have
\[\cost(Q)\leq 2\cdot \sum_{\text{component $c$}} \dist(r_c)\leq \frac{\eps}{4}\cdot\opt.\tag*{\qedhere}\]
\end{proof}

\begin{lemma}\label{lemma:analyzeITP}
Let $\Delta_1$ denote the cost in the ITP algorithm to connect the endpoints of all segments to the depot in Step~5.
We have $\Delta_1\leq \frac{\eps}{4}\cdot \opt$.
\end{lemma}

\begin{proof}
Let $n'$ denote the number of terminals in the tree $T$.
First, we show that the number of terminals on $t_{\TSP}$ is at most $4\alpha\cdot n'$.
Observe that the number of terminals on $t_{\TSP}$ is the overall removed demand in Step~4.
Consider any pseudo-tour $a_1\in A_1$ whose demand exceeds $k$.
Let $a_0$ denote the corresponding tour in $\OPT$.
By the construction, the demand of $a_1$ is at most the demand of $a_0$ plus the maximum demand of any element in $B$.
Since the demand of $a_0$ is at most $k$ and the demand of any element in $B$ is at most $2\alpha\cdot k$, the demand of $a_1$ is at most $k+2\alpha\cdot k$.
Let $a_2$ denote the corresponding tour after the correction of capacity in Step~4.
Since any element in $B$ has demand at most $2\alpha\cdot k$, the demand of $a_2$ is at least $k-2\alpha\cdot k$.
So the total removed demand from $a_1$ in Step~4 is at most $4\alpha\cdot k$.
There are at most $\frac{n'}{k}$ pseudo-tours $a_1\in A_1$ whose demands exceed $k$.
Summing over all those pseudo-tours, the overall removed demand in Step~4 is at most $\frac{n'}{k}\cdot 4\alpha\cdot k=4\alpha\cdot n'$.
Hence the number of terminals on $t_{\TSP}$ is at most $4\alpha\cdot n'$.

If $4\alpha\cdot n'\leq k$, then $t_{\TSP}$ is within the tour capacity, so the ITP algorithm is not applied and $\Delta_1=0$.
It remains to consider the case in which $4\alpha\cdot n'>k$.
By the construction in the ITP algorithm, every segment visits exactly $k$ terminals except possibly the last segment.
Thus the number $\ell_{\text{ITP}}$ of resulting tours in the ITP algorithm is
\begin{equation}
\label{eqn:L-ITP}
    \ell_{\text{IPT}}\leq \frac{4\alpha\cdot n'}{k}+1.
\end{equation}
Since $4\alpha\cdot n'>k$, we have $\ell_{\text{IPT}}<\frac{8\alpha\cdot n'}{k}$.
Since $\Delta_1\leq \ell_{\text{ITP}}\cdot 2\cdot D_{\max}$ and using \cref{def:bounded-distance} and the definition of $\alpha$ in the claim of \cref{thm:opt1}, we have
\[\Delta_1<\frac{8\alpha\cdot n'}{k}\cdot 2\cdot \left(\frac{1}{\eps}\right)^{\frac{1}{\eps}-1}\cdot D_{\min}<\frac{\eps}{2}\cdot \frac{n'}{k}\cdot D_{\min}.\]
On the other hand, the solution $\OPT$ consists of at least $\frac{n'}{k}$ tours, so $\opt\geq \frac{2n'}{k}\cdot D_{\min}.$
Therefore, $\Delta_1\leq \frac{\eps}{4}\cdot \opt$.
\end{proof}

\begin{lemma}
\label{lem:connect-small}
Let $\Delta_2$ denote the cost of the postprocessing (Step~6).
Then $\Delta_2\leq \frac{\eps}{2}\cdot \opt$.
\end{lemma}

\begin{proof}
For each leaf component $c$, the cost to connect the small subtours in $c$ to the depot in Step~6 is at most $2\cdot \dist(r_c)$; and for each internal component $c$, the cost to connect the small subtours in $c$ to the depot is at most $2\cdot\dist(r_c)+\cost(\spine_c)$.
Summing over all components, we have
\[\Delta_2\leq \sum_{\text{component }c} 2\cdot \dist(r_c) + \sum_{\text{internal component }c} \cost(\spine_c).\]
By \cref{lem:sum-dist-rc}, we have \[\sum_{\text{component }c} 2\cdot \dist(r_c)\leq \frac{\eps}{4}\cdot\opt\]
and
\[\sum_{\text{internal component }c} \cost(\spine_c)=\cost(Q)\leq \frac{\eps}{4}\cdot\opt.\]
Thus $\Delta_2\leq \frac{\eps}{2}\cdot\opt$.
\end{proof}

From \cref{cor:cost-Q,lemma:analyzeITP,lem:connect-small}, we have $\opt_1\leq \opt+\cost(Q)+\Delta_1+\Delta_2\leq (1+\eps)\cdot \opt$.
This completes the proof for \cref{thm:opt1}.

\section{Height Reduction}
\label{sec:simplifyingtree}
In this section, we transform the tree $T$ into a tree $\hat T$ so that $\hat T$ has $O_\eps (1)$ levels of components, see \cref{fig:transform-global}.
We assume that the tree $T$ has bounded distances.
To begin with, we partition the components according to the distances from their roots to the depot.

\begin{lemma}
\label{lem:tild-D-H-eps}
Let $\tilde D=\alpha\cdot\eps\cdot D_{\min}$.
Let $H_\eps=(\frac{1}{\eps})^{\frac{2}{\eps}+1}$.
For each $i\in[1,H_\eps]$, let $\mathcal{C}_i\subseteq \mathcal{C}$ denote the set of components $c\in \mathcal{C}$ such that $\dist(r_c)\in\left[(i-1)\cdot \tilde D, i\cdot \tilde D\right)$.
Then any component $c\in \mathcal{C}$ belongs to a set $\mathcal{C}_i$ for some $i\in [1,H_\eps]$.
\end{lemma}

\begin{proof}
Let $c\in \mathcal{C}$ be any component.
We have \[\dist(r_c)\leq D_{\max}< \left(\frac{1}{\eps}\right)^{\frac{1}{\eps}-1}\cdot D_{\min}=H_\eps\cdot \tilde D,\] where the second inequality follows from \cref{def:bounded-distance}, and the equality follows from the definition of $\alpha$ in \cref{thm:opt1} and the definitions of $\tilde D$ and $H_\eps$.
Thus there exists $i\in[1,H_\eps]$ such that
$c\in \mathcal C_{i}$.
\end{proof}

\begin{definition}[maximally connected sets and critical vertices]
\label{def:critical}
We say that a set of components $\tilde{\mathcal{C}}\subseteq \mathcal{C}_i$ is \emph{maximally connected} if the components in $\tilde{\mathcal{C}}$ are connected to each other and $\tilde{\mathcal{C}}$ is maximal within $\mathcal{C}_i$.
For a maximally connected set  of components $\tilde{\mathcal{C}}\subseteq \mathcal{C}_i$, we define the \emph{critical vertex} of $\tilde{\mathcal{C}}$ to be the root vertex of the component $c\in\tilde{\mathcal{C}}$ that is closest to the depot.
\end{definition}

\cref{fig:transform-global} (\cpageref{fig:transform-global}) is an example with three levels of components: $\mathcal{C}_1$, $\mathcal{C}_2$, and $\mathcal{C}_3$, indicated by different colors.
There are four maximally connected sets of components.
The critical vertices are represented by rectangular nodes.

\begin{algorithm}[h]
\caption{Construction of the tree $\hat T$ (see \cref{fig:transform-global}).}
\label{alg:hat-T}
\begin{algorithmic}[1]
\For{each $i\in[1,H_\eps]$}
    \For{each maximally connected set of components $\tilde{\mathcal{C}}\subseteq \mathcal{C}_i$}
        \State $z\gets$ critical vertex of $\tilde{\mathcal{C}}$
        \For{each component $c\in \tilde{\mathcal{C}}$}
            \State $\delta\gets r_c$-to-$z$ distance in $T$
            \State \emph{Split} the tree $T$ at the root vertex $r_c$ of the component $c$ \Comment{vertex $r_c$ is duplicated}
            \State Add an edge between the root of the component $c$ and $z$ with weight $\delta$
        \EndFor
    \EndFor
\EndFor
\State $\hat T\gets$the resulting tree
\end{algorithmic}
\end{algorithm}

Let $\hat T$ be the tree constructed in \cref{alg:hat-T}.
We observe that \cref{alg:hat-T} is in polynomial time, and \cref{fact:hat-T} follows from the construction.
We show in \cref{thm:opt2} that the optimal cost for $\hat T$ is increased by an $O(\eps)$ fraction compared with the optimal cost for $T$.

\begin{theorem}
\label{thm:opt2}
Consider the unit demand CVRP on the tree $\hat T$.
There exist dummy terminals and a solution $\OPT_2$ visiting all of the real and the dummy terminals, such that all of the following holds:
\begin{enumerate}
\item For each component $c$, there are at most $\frac{2\Gamma}{\alpha}+1$ tours visiting terminals in $c$;
\item For each component $c$ and each tour $t$ visiting terminals in $c$, the number of the terminals in $c$ visited by $t$ is at least $\alpha\cdot k$;
\item We have $\opt_2<(1+3\eps)\cdot\opt$, where $\opt$ denotes the optimal cost for the unit demand CVRP on the tree $T$.
\end{enumerate}
\end{theorem}

In the rest of the section, we prove \cref{thm:opt2}.

\subsection{Construction of $\OPT_2$}
\label{sec:construction-opt-2}

%To construct $\OPT_2$ on the tree $\hat T$, we start from $\OPT_1$ on the tree $T$.

Consider any tour $t$ in $\OPT_1$.
Let $U$ denote the set of terminals visited by $t$.\footnote{We assume that $t$ is a minimal tour in $T$ spanning all terminals in $U$.}
We define the tour $\hat t$  as the minimal tour in the tree $\hat T$ that spans all terminals in $U$, see \cref{fig:transform-global}.
Let $\OPT_2$ denote the set of the resulting tours on the tree $\hat T$ constructed from every tour $t$ in $\OPT_1$.
Then $\OPT_2$ is a feasible solution to the unit demand CVRP on $\hat T$.

\subsection{Analysis of $\OPT_2$}
\label{sec:analysis-opt2}
The first two properties in \cref{thm:opt2} follow from the construction and \cref{thm:opt1}.

In the rest of the section, we analyze the cost of $\OPT_2$.
\begin{lemma}
\label{lem:transform-t}
Let $t$ denote any tour in $\OPT_1$.
Let $\hat t$ denote the corresponding tour in $\OPT_2$.
Then $\cost(\hat t)\leq (1+\eps)\cdot\cost(t)$.
\end{lemma}

\begin{proof}
We follow the notation on $U$ from \cref{sec:construction-opt-2}.
Let $\mathcal{C}(U)$ denote the set of components $c\in \mathcal{C}$ that contains a (possibly spine) subtour of $\hat t$.
Observe that the cost of $\hat t$ consists of the following two parts:
\begin{enumerate}
    \item for each component $c\in \mathcal{C}(U)$, the cost of the subtour in $c$ from the tour $t$; we \emph{charge} that cost to the subtour in $t$;
    \item for each component $c\in \mathcal{C}(U)$, the cost of the edge $(r_c,z)$, where $z$ denotes the father vertex of $r_c$ in $\hat T$. Note that $z$ is a critical vertex on $\hat t$.
    We analyze that cost over all components $c\in \mathcal{C}(U)$ as follows.
\end{enumerate}
Let $Z\subseteq V$ denote the set of critical vertices $z\in V$ on $\hat t$.
For any critical vertex $z\in Z$, let $Y(z)$ denote the set of edges $(z,v)$ in the tree $\hat T$ such that $v$ is a child of $z$ and that the edge $(z,v)$ belongs to the tour $\hat t$.
The overall cost of the second part is the total cost of the edges in $Y(z)$ for all $z\in Z$.

Fix a critical vertex $z\in Z$.
Let $(z,v_1)$ denote the edge in $Y(z)$ such that $\dist(v_1)$ is minimized, breaking ties arbitrarily.
From the minimality of $\dist(v_1)$, the $z$-to-$v_1$ path in $T$ does not go through any component in $\mathcal{C}(U)$.
From the construction, the cost of the edge $(z,v_1)$ in $\hat T$ equals the cost of the $z$-to-$v_1$ path in $T$.
It is easy to see that the $z$-to-$v_1$ path in $T$ belongs to the tour $t$.
%Indeed, from the construction of $\hat T$, the terminals in the subtree in $\hat T$ rooted at $v_1$ are contained in the subtree in $T$ rooted at $v_1$.
%Since $(z,v_1)$ belongs to the tour $\hat t$, the subtree in $\hat T$ rooted at $v_1$ contains at least one terminal in $U$.
%Thus the subtree in $T$ rooted at $v_1$ contains at least one terminal in $U$,
%so the $z$-to-$v_1$ path in $T$ belongs to the tour $t$.
Indeed, tour $\hat{t}$ traverses the edge $(z,v_1)$ on its way to visit some terminals of $U$ in the subtree rooted at $v_1$. In order to visit the corresponding terminals in $T$, tour $t$ must traverse the $z$-to-$v_1$ path.
We \emph{charge} the cost of the edge $(z,v_1)$ in $\hat T$ to the $z$-to-$v_1$ path in $T$.
Next, we analyze the cost due to the other edges in $Y(z)$. Consider one such edge $(z,v)$.
%Next, we analyze the cost due to the remaining $|Y(z)|-1$ edges in $Y(z)$.
From the construction, there exists $i\in[1,H_\eps]$, such that both $\dist(z)$ and $\dist(v)$ belong to $\left[(i-1)\cdot \tilde D, i\cdot \tilde D\right)$.
%Observe that $z$ is on the $v$-to-$r$ path in $T$.
Thus the cost of the $z$-to-$v$ path in $T$ equals $\dist(v)-\dist(z)< \tilde D$, so the extra cost in $\hat t$ due to the edge $(z,v)$ is at most $2\cdot \tilde D$ (for both directions).
Therefore, the extra cost in $\hat t$ due to those $|Y(z)|-1$ edges in $Y(z)$ is at most $2\cdot \tilde D\cdot (|Y(z)|-1)$.

Summing over all vertices $z\in Z$, and observing that all charges are to disjoint parts of $t$, we have
\begin{equation}
\label{eqn:hat-t}
    \cost(\hat t)\leq \cost(t)+2\cdot \tilde D\cdot\sum_{z\in Z} (|Y(z)|-1).
\end{equation}

It remains to bound $\sum_{z\in Z}(|Y(z)|-1)$.
The analysis uses the following basic fact in trees.
\begin{fact}
\label{fact:degree}
Let $H$ be a tree with $L$ leaves.
For each vertex $u$ in $H$, let $m(u)$ denote the number of children of $u$ in $H$.
Then \[\sum_{u\in H} (m(u)-1)\leq L-1.\]
\end{fact}
We construct a tree $H$ as follows.
Starting from the tree spanning $U$ in $\hat T$, we contract vertices in each component $c\in \mathcal{C}(U)$ into a single vertex; let $H$ denote the resulting tree.
%By \cref{fact:degree} (and following the notations on $L$ and $m(\cdot)$ in \cref{fact:degree}), we have
%\[\sum_{z\in Z}(|Y(z)|-1)=\sum_{u\in H} (m(u)-1)\leq L-1.\]
It is easy to see that each leaf in $H$ corresponds to a component $c\in \mathcal{C}(U)$ that contains at least one terminal in $U$ (using the definition of $\mathcal{C}(U)$ and the fact that any descending component of $c$ do not belong to $\mathcal{C}(U)$).
From the second property of \cref{thm:opt2} (which follows from \cref{thm:opt1}), terminals in $U$ belong to at most $1/\alpha$ components. %, thus $L\leq 1/\alpha$.
Thus, by \cref{fact:degree} we have
 $$\sum_{z\in Z}(|Y(z)|-1)\leq (1/\alpha)-1.$$
Combined with \cref{eqn:hat-t}, we have \[\cost(\hat t)-\cost(t)<  2\cdot \tilde D\cdot (1/\alpha)=2\cdot \alpha\cdot \eps\cdot D_{\min}\cdot (1/\alpha)=2 \eps\cdot D_{\min},\] using the definition of $\tilde D$ in \cref{lem:tild-D-H-eps}.
Since $\cost(t)\geq 2\cdot D_{\min}$, the claim follows.
\end{proof}

Applying \cref{lem:transform-t} on each tour $t$ in $\OPT_1$ and summing, we have $\opt_2\leq (1+\eps)\cdot \opt_1$.
By \cref{thm:opt1},  $\opt_1\leq (1+\eps)\cdot \opt$, thus $\opt_2\leq (1+3\eps)\cdot \opt$.
This completes the proof of \cref{thm:opt2}.

\section{Adaptive Rounding on the Subtour Demands}\label{sec:analysisrounding}
In this section, we prove \cref{thm:opt3}.
We use the adaptive rounding to show that, in a near-optimal solution, the demands of the subtours at any critical vertex are from a set of $O_\eps(1)$ values.
This property enables us to later guess those values in polynomial time by a dynamic program (see \cref{sec:DP}).

\begin{theorem}
\label{thm:opt3}
Let $\beta=\frac{1}{4}\cdot \eps^{(\frac{4}{\eps}+1)}$.
Consider the unit demand CVRP on the tree $\hat T$.
There exist dummy terminals and a solution $\OPT_3$ visiting all of the real and the dummy terminals, such that all of the following holds:
\begin{enumerate}
\item For each component $c$, there are at most $\frac{2\Gamma}{\alpha}+1$ tours visiting terminals in $c$;
%\item For each component $c$ and each subtour $t$ in $c$, if $t$ visits terminals in $c$ and the number of the real terminals in $t$ is less than $\alpha\cdot k$, then $t$ contains dummy terminals such that the total number of the real and the dummy terminals in $t$ equals $\lceil \alpha\cdot k\rceil$;
\item For each critical vertex $z$, there exist $\frac{1}{\beta}$ integer values in $[\alpha\cdot k,k]$ such that the demands of the subtours at the children of $z$ are among these values;
\item We have $\opt_3<(1+4\eps)\cdot \opt$, where $\opt$ denotes the optimal cost for the unit demand CVRP on the tree $T$.
\end{enumerate}
\end{theorem}

\subsection{Construction of $\OPT_3$}
\label{sec:construction-opt3}

We construct the solution $\OPT_3$ by modifying the solution $\OPT_2$.

Let $I\subseteq V$ denote the set of vertices $v\in V$ that is either the root of a component or a critical vertex.
Consider any vertex $v\in I$ in the bottom up order.
Let $\OPT_2(v)$ denote the set of subtours at $v$ in $\OPT_2$.
We construct a set $A(v)$ of subtours at $v$ satisfying the following invariants:
\begin{itemize}
\item the subtours in $A(v)$ have a one-to-one correspondence with the subtours in $\OPT_2(v)$; and
\item the demand of each subtour of $A(v)$ is at most that of the corresponding subtour in $\OPT_2(v)$.
\end{itemize}
%We also maintain a set $B$ of additionals tours, which are due to the adaptive rounding.
The construction of $A(v)$ is according to one of the following three cases on $v$.

\paragraph{Case 1: $v$ is the root vertex $r_c$ of a leaf component $c$ in $\hat T$.}
Let $A(v)=\OPT_2(v)$.

\paragraph{Case 2: $v$ is the root vertex $r_c$ of an internal component $c$ in $\hat T$.}
For each subtour $a\in \OPT_2(v)$, if $a$ contains a subtour at the exit vertex $e_c$ of component $c$, letting $t$ denote this subtour and $t'$ denote the subtour in $A(e_c)$ corresponding to $t$, we replace the subtour $t$ in $a$ by the subtour $t'$.
Let $A(v)$ be the resulting set of subtours at $v$.

\paragraph{Case 3: $v$ is a critical vertex in $\hat T$.}
We apply the technique of the \emph{adaptive rounding}, previously used by Jayaprakash and Salavatipour~\cite{jayaprakash2021approximation} in their design of a QPTAS the tree CVRP.
The idea is to \emph{round up} the demands of the subtours at the children of $v$ so that the resulting demands are among $\frac{1}{\beta}$ values.

Let $r_1,\dots, r_m$ be the children of $v$ in $\hat T$.
For each subtour $a\in \OPT_2(v)$ and for each $i\in[1,m]$, if $a$ contains a subtour at $r_i$, letting $t$ denote this subtour and $t'$ denote the subtour in $A(r_i)$ corresponding to $t$, we replace $t$ in $a$ by $t'$.
Let $A_1(v)$ denote the resulting set of subtours at $v$.

Let $W_v$ denote the set of the subtours at the children of $v$ in $A_1(v)$, i.e., $W_v=A(r_1)\cup\dots\cup A(r_m)$.
If $|W_v|\leq \frac{1}{\beta}$, let $A(v)=A_1(v)$.
In the following, we consider the non-trivial case when $|W_v|>\frac{1}{\beta}$.
We sort the subtours in $W_v$ in non-decreasing order of their demands, and partition these subtours into $\frac{1}{\beta}$ groups of equal cardinality.\footnote{We add empty subtours to the first groups if needed in order to achieve equal cardinality among all groups.}
We \emph{round} the demands of the subtours in each group to the maximum demand in that group.
The demand of a subtour is increased to the rounded value by adding \emph{dummy} terminals at the children of $v$.
We rearrange the subtours in $W_v$ as follows.
\begin{itemize}
\item Each subtour $t\in W_v$ in the last group is discarded, i.e., detached from the subtour in $A_1(v)$ to which it belongs.
\item Each subtour $t\in W_v$ in other groups is associated in a one-to-one manner to a subtour $t'\in W_v$ in the next group.
Letting $a$ (resp.\ $a'$) denote the subtour in $A_1(v)$ to which $t$ (resp.\ $t'$) belongs, we detach $t$ from $a$ and reattach $t$ to $a'$.
\end{itemize}
Let $A(v)$ be the set of the resulting subtours at $v$ after the rearrangement for all $t\in W_v$.

\paragraph{}
For each subtour $t$ that is discarded in the construction, we complete $t$ into a separate tour by adding the connection (in both directions) to the depot.
Let $B$ denote the set of these newly created tours.
Let $\OPT_3=A(r)\cup B$.

It is easy to see that $\OPT_3$ is a feasible solution to the unit demand CVRP, i.e., each tour in $\OPT_3$ is connected and visits at most $k$ terminals, and each terminal is covered by some tour in $\OPT_3$.

\subsection{Analysis of $\OPT_3$}\label{subsection:analysistreeandDP}
From the construction, in any component $c\in\mathcal{C}$, the non-spine subtours  in $\OPT_3$ are the same as those in $\OPT_2$.
From \cref{thm:opt2}, we obtain the first property in \cref{thm:opt3}, and in addition, each subtour at a child of a critical vertex in $\OPT_2$ has demand at least $\alpha\cdot k$.
The second property of the claim follows from the construction of $\OPT_3$.

It remains to analyze the cost of $\OPT_3$.
Let $\Delta=\opt_3-\opt_2$.
Observe that $\Delta$ is due to adding connections to the depot to create the tours in the set $B$.

Fix any $i\in[1,H_\eps]$.
Let $Z\subseteq V$ denote the set of vertices $v\in V$ such that $v$ is the critical vertex of a maximally connected component $\tilde{\mathcal{C}}\subseteq \mathcal{C}_i$.
For any $v\in Z$, we analyze the number of discarded subtours in the set $W_v$ defined in \cref{sec:construction-opt3}.
If $|W_v|\leq \frac{1}{\beta}$, there is no discarded subtour in $W_v$; if $|W_v|>\frac{1}{\beta}$, the number of discarded subtours in $W_v$ is $\lceil\beta\cdot |W_v|\rceil<\beta\cdot |W_v|+1<2\beta\cdot |W_v|$.
Let $W$ denote the disjoint union of $W_v$ for all vertices $v\in Z$.
Thus $W$ contains at most $2\beta\cdot |W|$ discarded subtours.
Let $\Delta_i$ denote the cost to connect the discarded subtours in $W$ to the depot.
We have
\begin{equation}
    \label{eqn:Delta-i}
    \Delta_i\leq 2\beta\cdot |W|\cdot 2\cdot D_{\max}< \frac{1}{2}\cdot \eps^{(\frac{4}{\eps}+1)}\cdot |W|\cdot 2\cdot \left(\frac{1}{\eps}\right)^{\frac{1}{\eps}-1}\cdot D_{\min}= \frac{\eps}{H_\eps}\cdot\alpha\cdot |W|\cdot D_{\min},
\end{equation}
where the second inequality follows from the definition of $\beta$ (\cref{thm:opt3}) and \cref{def:bounded-distance}, and the equality follows from the definitions of $\alpha$ (\cref{thm:opt1}) and of $H_\eps$ (\cref{lem:tild-D-H-eps}).
From the second property of the claim, each subtour in $W$ has demand at least $\alpha\cdot k$, so there are at least $\alpha\cdot |W|$ tours in $\OPT_3$.
Any tour in $\OPT_3$ has cost at least $2\cdot D_{\min}$, so we have
\begin{equation}
    \label{eqn:opt-3}
    \opt_3\geq 2\cdot \alpha\cdot |W|\cdot D_{\min}.
\end{equation}
From \cref{eqn:Delta-i,eqn:opt-3}, we have
\[\Delta_i< \frac{\eps}{2\cdot H_\eps}\cdot\opt_3.\]

Summing over all integers $i\in[1,H_\eps]$, we have
$\Delta=\sum_i\Delta_i\leq \frac{\eps}{2}\cdot\opt_3$.
Thus
\[\opt_3\leq \frac{2}{2-\eps}\cdot\opt_2.\]
By \cref{thm:opt2}, $\opt_2\leq (1+3\eps)\cdot\opt$.
Therefore, $\opt_3\leq (1+4\eps)\cdot\opt$.
This completes the proof of \cref{thm:opt3}.

\section{Dynamic Programming}\label{sec:DP}
In this section, we show \cref{thm:DP}.
In our dynamic program, we consider all feasible solutions on the tree~$\hat T$ satisfying the properties of $\OPT_3$ in \cref{thm:opt3}, and we output the solution with minimum cost.
The first property of \cref{thm:opt3} is used in \cref{sec:local} in the computation of solutions inside components, and the second property of \cref{thm:opt3} is used in \cref{sec:subtree} in the computation of solutions in the subtrees rooted at critical vertices.
These two properties ensure the polynomial running time of the dynamic program.
The cost of the output solution is at most the cost of $\OPT_3$, which is at most $(1+4\eps)$ times the optimal cost on the tree
$T$ by the third property of \cref{thm:opt3}.

%%%%%%%%%%%%%%%%%%%%%%%%%%%%%%
\subsection{Local Configurations}
\label{sec:local}

In this subsection, we compute values at \emph{local configurations} (\cref{def:local-config}), which are solutions restricted locally to a component.
We require that the terminals in any component are visited by at most $\frac{2\Gamma}{\alpha}+1$ tours, using the first property in \cref{thm:opt3}.
Thus the number of local configurations is polynomially bounded.

\begin{definition}[local configurations]
\label{def:local-config}
Let $c\in \mathcal{C}$ be any component.
A \emph{local configuration} $(v,A)$ is defined by a vertex $v\in c$ and a list $A$ of $\ell(A)$ pairs
$(s_1,b_1),(s_2,b_2),\ldots , (s_{\ell(A)},b_{\ell(A)})$ such that
\begin{itemize}
\item
$\ell(A)\leq \frac{2\biggamma}{\alpha}+1$;
\item
for each $i\in[1, \ell(A)]$,  $s_i$ is an integer in $[0,k]$ and $b_i\in \{``\hbox{passing}",``\hbox{ending}"\}$.\footnote{If $c$ is a leaf component, then $b_i$ is ``ending'' for each $i$. For technical reasons due to the exit vertex, we allow $s_i$ to take the value of $0$.}
\end{itemize}
When $v=r_c$, the local configuration $(r_c,A)$ is also called a \emph{local configuration in the component $c$}.
\end{definition}

The \emph{value} of a local configuration $(v,A)$, denoted by $f(v,A)$, equals the minimum cost of a collection of $\ell(A)$ subtours in the subtree of $c$ rooted at $v$, each subtour starting and ending at $v$, that together visit all of the terminals of the subtree of $c$ rooted at $v$, where the $i$-th subtour visits $s_i$ terminals and that $b_i=\text{``passing''}$ if and only if the $i$-th subtour visits $e_c$.
%Note that $s_1+s_2+\cdots +s_{\ell(A)}$ is equal to the number of terminals in the subtree of $c$ rooted at $v$.

Let $v$ be any vertex in $c$.
We compute the function $f(v,\cdot)$ according to one of the three cases.

\paragraph{Case 1: $v$ is the exit vertex $e_c$ of the component $c$.}
For each $\ell\in [0,2\biggamma/\alpha+1]$, letting $A$ denote the list consisting uniquely of $\ell$ identical pairs of $(0,\text{``passing''})$, we set $f(v,A)=0$; for the remaining lists $A$, we set $f(v,A)=+\infty$.

\paragraph{Case 2: $v$ is a leaf vertex of the tree $\hat T$.}
From \cref{sec:preliminary} and the construction of $\hat T$, the leaf vertices in $\hat T$ are the same as the terminals in $\hat T$.
Thus $v$ is a terminal in $\hat T$.
For the list $A$ consisting of a single pair of $(1,\text{``ending''})$, we set $f(v,A)=0$; for the remaining lists $A$, we set $f(v,A)=+\infty$.

\paragraph{Case 3: $v$ is a non-leaf vertex of the tree $\hat T$ and $v\neq e_c$.}
Let $v_1$ and $v_2$ be the two children of $v$.
%\footnote{We assume without loss of generality that a non-leaf vertex $v$ has two children. If $v$ has only one child, then we may eliminate $v$ by contracting the two edges incident to $v$ into a single edge.}
We say that the local configurations $(v_1,A_1)$, $(v_2,A_2)$, and $(v,A)$ are \emph{compatible} if there is a partition $\mathcal{P}$ of $A_1\cup A_2$ into parts, each part consisting of one or two pairs, and a one-to-one correspondence between every part in $\mathcal{P}$ and every pair in $A$ such that:
\begin{itemize}
    \item a part in $\mathcal{P}$ consisting of one pair $(s^{(1)},b^{(1)})$ corresponds to a pair $(s,b)$ in $A$ if and only if $s^{(1)}=s$ and $b^{(1)}=b$;
    \item a part in $\mathcal{P}$ consisting of two pairs $(s^{(1)},b^{(1)})$ and $(s^{(2)},b^{(2)})$ corresponds to a pair $(s,b)$ in $A$ if and only if $s=s^{(1)}+s^{(2)}$ and
    \[b=\begin{cases}
    \text{``passing'',} & \text{if $b^{(1)}$ is ``passing'' or $b^{(2)}$ is ``passing''},\\
    \text{``ending'',} & \text{if $b^{(1)}$ is ``ending'' and $b^{(2)}$ is ``ending''.}\\
    \end{cases}
    \]
\end{itemize}
We set
\[f(v,A)=\min \big\{ f(v_1,A_1)+f(v_2,A_2)+2\cdot\ell(A_1)\cdot w(v,v_1)+2\cdot\ell(A_2)\cdot w(v,v_2)\big\},\]
where the minimum is taken over all local configurations $(v_1,A_1)$ and $(v_2,A_2)$ that are compatible with $(v,A)$.

\paragraph{}The algorithm is very simple. See \cref{alg:local}.
\begin{algorithm}
\caption{Computation for local configurations in a component $c$.}
\label{alg:local}
\begin{algorithmic}[1]
\For{each vertex $v\in c$ and each list $A$}
    \State $f(v,A)=+\infty$.
\EndFor
    \For{each leaf vertex $v$ of $c$}\Comment{Cases 1 \& 2}
    \State Initialize $f(v,\cdot)$
    \EndFor
    \For{each non-leaf vertex $v$ of $c$ in the bottom-up order }\Comment{Case 3}
        \State Let $v_1$ and $v_2$ denote the two children of $v$
        \For{each local configurations $(v_1,A_1)$, $(v_2,A_2)$, and $(v,A)$}
	        \If{$(v_1,A_1)$, $(v_2,A_2)$, and $(v,A)$ are compatible}
	        \State $f(v,A)\gets \min(f(v,A),f(v_1,A_1)+f(v_2,A_2)+2\cdot\ell(A_1)\cdot w(v,v_1)+2\cdot\ell(A_2)\cdot w(v,v_2))$
            \EndIf
        \EndFor
    \EndFor
\State \Return $f(r_c,\cdot)$
\end{algorithmic}
\end{algorithm}

\paragraph{Running time.}
For each vertex $v$, since $\ell(A)=O_\eps(1)$, the number of local configurations $(v,A)$ is $n^{O_\eps(1)}$. For fixed $(v_1,A_1),(v_2, A_2)$, and $(v,A)$, there are $O_\eps(1)$ partitions of $A_1\cup A_2$ into parts, so checking compatibility takes time $O_\eps(1)$.
Thus the running time to compute the values at local configurations in a component $c$ is $n^{O_\eps(1)}$.
Since the number of components $c\in \mathcal{C}$ is at most $n$, the overall running time to compute the local configurations in all components $c\in\mathcal{C}$ is $n^{O_\eps(1)}$.

%=====================================================================

\subsection{Subtree Configurations}
\label{sec:subtree}

In this subsection, we combine local configurations in the bottom up order to obtain \emph{subtree configurations} (\cref{def:subtree-config}), which are solutions restricted to subtrees of $\hat T$.
The number of subtree configurations is polynomially bounded.
When the subtree equals the tree $\hat T$, we obtain the entire solution to the CVRP.

\begin{definition}
\label{def:subtree-config}
A \emph{subtree configuration} $(v,A)$ is defined by a vertex $v$ and a list $A$ consisting of $\ell(A)$ pairs $(\tilde s_1,n_1)$, $(\tilde s_2,n_2),\ldots,(\tilde s_\ell,n_{\ell(A)})$ such that
\begin{itemize}
\item $v$ belongs to the set $I$ (defined in \cref{sec:construction-opt3}); in other words, $v$ is either the root of a component or a critical vertex;
\item
$\ell(A)\leq \left(\frac{1}{\beta}\right)^{\frac{1}{\alpha}}+\frac{2\biggamma}{\alpha}+1$ if $v$ is the root of a component, and $\ell(A)\leq \left(\frac{1}{\beta}\right)^{\frac{1}{\alpha}}$ if $v$ is a critical vertex;
\item
for each $i\in[1, \ell(A)]$,  $\tilde s_i$ is an integer in $[0,k]$ and $n_i$ is an integer in $[0,n]$.
\end{itemize}
\end{definition}
The \emph{value} of the subtree configuration $(v,A)$, denoted by $g(v,A)$, is the minimum cost of a collection of $\ell(A)$ subtours in the subtree of $\hat T$ rooted at $v$, each subtour starting and ending at $v$, that together visit all of the real terminals of the subtree rooted at $v$, such that $n_i$ subtours visit $\tilde s_i$ (real and dummy) terminals each.

To compute the values of subtree configurations, we consider the vertices $v\in I$ in the bottom up order.
See \cref{fig:DP} (\cpageref{fig:DP}).
For each vertex $v\in I$ that is the root of a component, we compute the values $g(v,\cdot)$ using the algorithm in \cref{sec:subtree-configuration-root}; and for each vertex $v\in I$ that is a critical vertex, we compute the values $g(v,\cdot)$ using the algorithm in \cref{sec:subtree-configuration-critical}.

\subsubsection{Subtree Configurations at the Root of a Component}
\label{sec:subtree-configuration-root}
In this subsection, we compute the values of the subtree configurations at the root $r_c$ of a component $c$.
From \cref{sec:local}, we have already computed the values of the local configurations in the component $c$.

If $c$ is a leaf component, the local configurations in $c$ induce the subtree configurations at $r_c$, in which $\tilde s_i=s_i$ and $n_i=1$ for all $i$.
Thus we obtain the values of the subtree configurations at $r_c$.

In the following, we consider the case when $c$ is an internal component.
We observe that the exit vertex $e_c$ of the component $c$ is a critical vertex.
Thus the values of subtree configurations at $e_c$ have already been computed using \cref{alg:subtree-configuration-critical} in \cref{sec:subtree-configuration-critical} according to the bottom up order of the computation.
To compute the value of a subtree configuration at $r_c$, we combine a subtree configuration at $e_c$ and a local configuration in $c$, in the following way.

Let $A_e=((\tilde s_1,n_1),(\tilde s_2,n_2),\ldots,(\tilde s_{\ell_e},n_{\ell_e}))$ be the list from a subtree configuration $(e_c,A_e)$.
Let $A_c=((s_1,b_1),(s_2,b_2),\dots,(s_{\ell_c},b_{\ell_c}))$ be the list from a local configuration $(r_c,A_c)$.
To each $i\in[1,\ell_c]$ such that $b_i$ is ``passing'', we associate $s_i$ with $\tilde s_j$ for some $j\in[1,\ell_e]$ with the constraints that $s_i+\tilde s_j\leq k$ (to guarantee that when we combine the two subtours, the result respects the capacity constraint) and that for each $j\in[1,\ell_e]$ at most $n_j$ elements are associated to $\tilde s_j$ (because in the subtree rooted at $e_c$ we only have $n_j$ subtours of demand $\tilde s_j$ at our disposal).
As a result, we obtain the list $A$ of a subtree configuration $(r_c,A)$ as follows:
\begin{itemize}
    \item For each association $(s_i,\tilde s_j)$, we put in $A$ the pair $(s_i+\tilde s_j,1)$.
    \item For each pair $(\tilde s_j,n_j)\in A_e$, we put in $A$ the pair $(\tilde s_j,n_j-(\text{number of $s_i$'s associated to $\tilde s_j$}))$.
    \item For each pair $(s_i,``\text{ending}")\in A_c$, we put in $A$ the pair $(s_i,1)$.
\end{itemize}
From the construction, $\ell(A)\leq \ell(A_e)+\ell(A_c)$.
Since $e_c$ is a critical vertex, $\ell(A_e)\leq \left(\frac{1}{\beta}\right)^{\frac{1}{\alpha}}$ by \cref{def:subtree-config}.
From \cref{def:local-config}, $\ell(A_c)\leq \frac{2\biggamma}{\alpha}+1$.
Thus $\ell(A)\leq \left(\frac{1}{\beta}\right)^{\frac{1}{\alpha}}+\frac{2\biggamma}{\alpha}+1$ as claimed in \cref{def:subtree-config}.

Next, we compute the cost of the combination of $(e_c,A_e)$ and $(r_c,A_c)$; let $x$ denote this cost.
For any subtour $t$ at $e_c$ that is not associated to any non-spine passing subtour in the component $c$, we pay an extra cost to include the spine subtour of the component $c$, which is combined with the subtour $t$.
The number of times that we include the spine subtour of $c$ is the number of subtours at $e_c$ minus the number of passing subtours in $A_c$, which is $\sum_{j\leq \ell_e} n_j- \sum_{i\leq \ell_c} \mathbbm{1} \left[b_i \text{ is ``passing''}\right]$.
Thus we have
\begin{equation}  \label{eq:costcomputation}
x=f(r_c,A_c)+g(e_c,A_e)+\cost(\spine_c)\cdot\left(\bigg(\sum_{j\leq \ell_e} n_j\bigg) - \bigg(\sum_{i\leq \ell_c} \mathbbm{1} \left[b_i \text{ is ``passing''}\right]\bigg)\right).\end{equation}

The algorithm is described in \cref{alg:subtree-configuration-root}.

\begin{algorithm}
\caption{Computation for subtree configurations at the root of a component $c$.}
\label{alg:subtree-configuration-root}
\begin{algorithmic}[1]
%Initialization
\For{each list $A$}
    \State $g(r_c,A)=+\infty$.
\EndFor
%Computation
\For{each subtree configuration $(e_c,A_e)$ and each local configuration $(r_c,A_c)$}
        \For{each way to combine $(e_c,A_e)$ and $(r_c,A_c)$}
            \State $A\gets$ the resulting list
            \State $x\gets$ the cost computed in Equation~(\ref{eq:costcomputation})
            \State $g(r_c,A)\gets \min (g(r_c,A),x)$.
        \EndFor
\EndFor
\State \Return $g(r_c,\cdot)$
\end{algorithmic}
\end{algorithm}

\paragraph{Running time.}
The number of subtree configurations $(e_c,A_e)$ and the number of local configurations $(r_c,A_c)$ are both $n^{O_\eps(1)}$.
For fixed $(e_c,A_e)$ and $(r_c,A_c)$, the number of ways to combine them is $O_\eps(1)$.
Thus the running time of the algorithm is $n^{O_\eps(1)}$.

%========================================
\subsubsection{Subtree Configurations at a Critical Vertex}

\label{sec:subtree-configuration-critical}
In this subsection, we compute the values of the subtree configurations at a critical vertex.

Let $z$ denote any critical vertex.
By Property~2 in \cref{thm:opt3}, there exists a set $X$ of $\frac{1}{\beta}$ integer values in $[\alpha\cdot k,k]$ such that the demands of the subtours in $\OPT_3$ at the children of $z$ are among the values in $X$.
%Again by \cref{thm:opt3}, each tour in $\OPT_3$ visits terminals in at most $\frac{1}{\alpha}$ components.
Thus the demand of a subtour at $z$ in $\OPT_3$ is the sum of at most $\frac{1}{\alpha}$ values in $X$.
Therefore, the number of distinct demands of the subtours at $z$ in $\OPT_3$ is at most $(\frac{1}{\beta})^{\frac{1}{\alpha}}=O_\eps(1)$.

To compute a solution satisfying Property~2 in
\cref{thm:opt3}, a difficulty arises since the set $X$ is unknown.
Our approach is to enumerate all sets $X$ of $\frac{1}{\beta}$ integer values in $[\alpha\cdot k,k]$, compute a solution with respect to each set $X$, and return the best solution found.
Unless explicitly mentioned, we assume in the following that the set $X$ is fixed.

\begin{definition}[sum list]
A \emph{sum list} $A$ consists of $\ell(A)$ pairs
$(s_1,n_1)$, $(s_2,n_2)$, \dots, $(s_{\ell(A)},n_{\ell(A)})$ such that
\begin{enumerate}
    \item
$\ell(A)\leq (\frac{1}{\beta})^{\frac{1}{\alpha}}$;
\item
For each $i\in[1,\ell(A)]$, $s_i\in[\alpha\cdot k,k]$ is the sum of a multiset of values in $X$ and $n_i$ is an integers in $[0,n]$.
\end{enumerate}
\end{definition}
We require that in any subtree configuration $(z,A)$, the list $A$ is a sum list.

Let $r_1,r_2,\ldots ,r_m$ be the children of $z$.
For each $i\in[1,m]$, let \[A_i=( (s_1^{(i)},n_1^{(i)}),(s_2^{(i)},n_2^{(i)}),\ldots , (s_{\ell_i}^{(i)},n_{\ell_i}^{(i)}))\] denote the list in a subtree configuration $(r_i,A_i)$.
We \emph{round} the list $A_i$ to a list \[\overline{A_i}=( (\overline{s_1^{(i)}},n_1^{(i)}),\overline{(s_2^{(i)}},n_2^{(i)}),\ldots , (\overline{s_{\ell_i}^{(i)}},n_{\ell_i}^{(i)})),\] where $\overline{x}$  denotes the smallest value in $X$ that is greater than or equal to $x$, for any integer value $x$.
The rounding is represented by adding $\overline{x}-x$ \emph{dummy} terminals at vertex $r_i$ to each subtour initially consisting of $x$ terminals.
Let $\mathcal{S}\subseteq [1,k]$ denote a multiset such that for each $i\in[1,m]$ and for each $j\in[1,\ell_i]$, the multiset $\mathcal{S}$ contains $n_j^{(i)}$ copies of $\overline{s_j^{(i)}}$.

\begin{definition}[compatibility]
A multiset $\mathcal{S}\subseteq [1,k]$ and a sum list $( (s_1,n_1),(s_2,n_2),\ldots , (s_{\ell},n_{\ell}))$ are \emph{compatible} if there is a partition of $\mathcal{S}$ into $\sum_i n_i$ parts and a correspondence between the parts of the partition and the values $s_i$'s, such that for each $s_i$, there are $n_i$ associated parts, and for each of those parts, the elements in that part sum up to $s_i$.
\end{definition}

Let $A=( (s_1,n_1),(s_2,n_2),\ldots , (s_{\ell(A)},n_{\ell(A)}))$ be a sum list.
The value $g(z,A)$ of the subtree configuration $(z,A)$ equals the minimum, over all sets $X$ and all subtree configurations $\{(r_i,A_i)\}_{1\leq i\leq m}$ such that $\mathcal{S}$ and $A$ are compatible, of
\begin{equation}
\label{eqn:g-z-A}
\sum_{i=1}^m g(r_i,A_i)+ 2\cdot n(A_i)\cdot w(r_i,z),
\end{equation}
where $n(A_i)$ denotes $\sum_j n_j^{(i)}$.
We note that $n_1 s_1+n_2 s_2+\cdots +n_{\ell(A)} s_{\ell(A)}$ is equal to the number of (real and dummy) terminals in the subtree rooted at $z$.

\begin{algorithm}[t]
\caption{Computation for subtree configurations at a critical vertex $z$.}
\label{alg:subtree-configuration-critical}
\begin{algorithmic}[1]
    \For{each list $A$}
        \State $g(z,A)\gets+\infty$
    \EndFor
    \For{each set $X$ of $\frac{1}{\beta}$ integer values in $[\alpha\cdot k,k]$}
        \For{each $i\in [0,m]$ and each list $A$}
            \State $\DP_i(A)\gets +\infty$
        \EndFor
        \State $\DP_0(\emptyset)\gets 0$
        \For{each $i\in[1,m]$}
            \For{each subtree configuration $(r_i,A_i)$}
                \State $\overline{A_i}\gets round(A_i)$
                \For{each sum list $A_{\leq i-1}$}
                    \For{each way to combine $A_{\leq i-1}$ and $\overline{A_i}$}
                        \State $A_{\leq i}\gets$ the resulting sum list
                        \State $x\gets \DP_{i-1}(A_{\leq i-1})+g(r_i,A_i)+2\cdot n(A_i)\cdot w(r_i,z)$
                        \State $\DP_i(A_{\leq i})\gets \min(\DP_i(A_{\leq i}),x)$
                    \EndFor
                \EndFor
            \EndFor
    \EndFor
    \For{each list $A$}
        \State $g(z,A)\gets \min(g(z,A),\DP_m(A))$
    \EndFor
\EndFor
\State \Return $g(z,\cdot)$
\end{algorithmic}
\end{algorithm}

Fix any set $X$ of $\frac{1}{\beta}$ integer values in $[\alpha\cdot k,k]$.
We show how to compute the minimum cost of~\cref{eqn:g-z-A} over all subtree configurations $\{(r_i,A_i)\}_{1\leq i\leq m}$ such that $\mathcal{S}$ and $A$ are compatible.
For each $i\in[1,m]$ and for each subtree configuration $(r_i,A_i)$, the value $g(r_i,A_i)$ has already been computed using the algorithm in \cref{sec:subtree-configuration-root}, according to the bottom up order of the computation.
We use a dynamic program that scans $r_1,\dots, r_m$ one by one: those are all siblings, so here the reasoning is not bottom-up but left-right.
Fix any $i\in[1,m]$.
Let $\mathcal{S}_i\subseteq[1,k]$ denote a multiset such that for each $i'\in[1,i]$ and for each $j\in[1,\ell_{i'}]$, the multiset $\mathcal{S}_i$ contains $n_j^{(i')}$ copies of $\overline{s_j^{(i')}}$.
We define a dynamic program table $\DP_i$.
The value $\DP_i(A_{\leq i})$ at a sum list $A_{\leq i}$ equals the minimum, over all subtree configurations $\{(r_{i'},A_{i'})\}_{1\leq i'\leq i}$ such that $\mathcal{S}_i$ and $A_{\leq i}$ are compatible, of   \[\sum_{i'=1}^i g(r_{i'},A_{i'})+ 2\cdot n(A_{i'})\cdot w(r_{i'},z).\]
When $i=m$, the values $\DP_m(\cdot)$ are those that we are looking for.
It suffices to fill in the tables $\DP_1,\DP_2,\dots,\DP_m$.

To compute the value $\DP_i$ at a sum list $A_{\leq i}$, we use the value $\DP_{i-1}$ at a sum list $A_{\leq i-1}$ and the value $g(r_i,A_i)$ of a subtree configuration $(r_i,A_i)$.
Let $A_{\leq i-1}=( (\hat s_1,\hat n_1),(\hat s_2,\hat n_2),\ldots , (\hat s_{\ell},\hat n_{\ell}))$.
We combine $A_{\leq i-1}$ and $\overline{A_i}$ as follows.
For each $p\in [1,\ell]$ and each $j\in[1,\ell_i]$ such that $\hat s_p+\overline{s_{j}^{(i)}}\leq k$, we observe that $\hat s_p+\overline{s_{j}^{(i)}}$ is the sum of a multiset of values in $X$.
We create $n_{p,j}$ copies of the association of $(\hat s_p,\overline{s_{j}^{(i)}})$, where $n_{p,j}\in [0,n]$ is an integer variable that we enumerate in the algorithm.
We require that for each $p\in[1,\ell]$, $\sum_j n_{p,j}\leq \hat n_p$; and for each $j\in[1,\ell_i]$, $\sum_p n_{p,j} \leq n_j^{(i)}$.
The resulting sum list $A_{\leq i}$ is obtained as follows.
\begin{itemize}
\item For each association $(\hat s_p,\overline{s_{j}^{(i)}})$, we put in $A_{\leq i}$ the pair $(\hat s_p+\overline{s_{j}^{(i)}},n_{p,j})$.
\item For each pair $(\hat s_p,\hat n_p)\in A_{\leq i-1}$, we put in $A_{\leq i}$ the pair $(\hat s_p,\hat n_p-\sum_j n_{p,j})$.
\item For each pair $(\overline{s_j^{(i)}},n_j^{(i)})\in \overline{A_i}$, we put in $A_{\leq i}$ the pair $(\overline{s_j^{(i)}},n_j^{(i)}-\sum_p n_{p,j})$.
\end{itemize}

The algorithm is described in \cref{alg:subtree-configuration-critical}.

\paragraph{Running time.} Since the numbers of sets $X$, of subtree configurations, of sum lists, and of ways to combine them, are each $n^{O_\eps(1)}$, the running time of the algorithm is $n^{O_\eps(1)}$.

\section{Reduction to Bounded Distances}\label{sec:preprocessing}

In this section, we prove \cref{thm:reduction-bounded-dist}.
We reduce the tree CVRP with general distances to the tree CVRP with bounded distances.
The reduction holds for the unit demand version, the splittable version, and the unsplittable version of the tree CVRP.

\subsection{Algorithm}
\label{sec:general-dist-algo}

For any subset $U\subseteq V'$ of terminals, a \emph{subproblem on $U$} is an instance to the tree CVRP, in which the tree is $T=(V,E)$ and the set of terminals is $U$.
To simplify the presentation, we assume that $1/\eps$ is an integer.
%A \emph{layer}  of the tree $T$, parameterized by an integer $i$, is the set of vertices whose distance to the root is in .
For each integer $i\in \mathbb{Z}$, let $U_i\subseteq V'$ denote the set of terminals $v\in V'$ such that $\dist(v)\in [(1/\eps)^i, (1/\eps)^{i+1})$.

Choose an integer $i_0\in [0,(1/\eps)-1]$ uniformly at random.
For each integer $j\in \mathbb{Z}$, let $Y_j=U_{(1/\eps)\cdot j+i_0}$, and let $Z_j\subseteq V'$ denote the union of $U_i$ for $i=(1/\eps)\cdot j +i_0+1, (1/\eps)\cdot j+i_0+2,\ldots ,(1/\eps)\cdot(j+1) + i_0-1$.
Let $W$ denote the collection of the non-empty sets $Y_j$'s and the non-empty sets $Z_j$'s.
Note that $W$ is a partition of the terminals in $V'$.

Let $\mathcal{A}$ denote any polynomial time $\rho$-approximation algorithm for the tree CVRP with bounded distances from the assumption in \cref{thm:reduction-bounded-dist}.
Consider any set $U\in W$.
From the construction, we have
\[\frac{\max_{v\in U}\dist(v)}{\min_{v\in U}\dist(v)}< \left(\frac{1}{\eps}\right)^{\frac{1}{\eps}-1}.\]
Thus the subproblem on $U$ has bounded distances.
We apply the algorithm $\mathcal{A}$ on the subproblem on $U$ to obtain a solution $\SOL(U)$.

Let $\SOL=\bigcup_{U\in W} \SOL(U)$.

It is easy to see that $\SOL$ is a feasible solution to the CVRP.
Since the number of subproblems is at most $n$ and each subproblem is solved in polynomial time by $\mathcal{A}$, the overall running time
is polynomial.

\begin{remark}
The algorithm can be derandomized by enumerating all of the $1/\eps$ values of $i_0$, and returning the best solution found.
\end{remark}

\subsection{Analysis}

We analyze the cost of the solution $\SOL$.

For any subset $U\subseteq V'$ of terminals, let $\opt(U)$ denote the optimal value for the subproblem on $U$.
For each set $U\in W$, $\SOL(U)$ is a $\rho$-approximate solution to the subproblem on $U$.
Thus we have
\begin{equation}
\label{eqn:SOL}
\cost(\SOL)=\sum_{U\in W} \cost(\SOL(U))\leq \rho\cdot \sum_{U\in W} \opt(U)=\rho\cdot \left(\sum_{j} \opt(Y_j)+\sum_j \opt(Z_j)\right).
\end{equation}
In the following, we bound $\sum_{j} \opt(Y_j)$ and $\sum_j \opt(Z_j)$.

For each $i\in \mathbb{Z}$, define $F_i\subseteq E$ to be
\[F_i=\big\{(u,v)\in E \mid \min(\dist(u),\dist(v))\in\left[(1/\eps)^i,(1/\eps)^{i+1}\right)\big\}.\]
We assume without loss of generality that, for each $i\in \mathbb{Z}$ and for each $(u,v)\in F_i$, $\max(\dist(u),\dist(v))\leq (1/\eps)^{i+1}$. Indeed, if $\max(\dist(u),\dist(v))> (1/\eps)^{i+1}$, we may replace the edge $(u,v)$ by a path of edges whose total weight equals $w(u,v)$, and such that each edge on that path satisfies the assumption.

Let $F_{\leq i}$ denote the union of $F_{i'}$ for all $i'\leq i$.

\begin{lemma}\label{lemma:samelayer}
$\sum_i \opt(U_i)\leq 2(1+\eps)\cdot \opt$.
\end{lemma}

\begin{proof}
Consider any tour $t$ in $\OPT$.
Let $i^*$ be the maximum index $i$ such that $t\cap U_i\neq \emptyset$.
Let $t_{i^*}$ denote the tour obtained by pruning $t$ so that it visits only the root and the terminals in $U_{i^*}$.
Let $t_{\leq i^*-1}$ denote the tour obtained by pruning $t$ so that it visits only the root and the remaining terminals, i.e.,\ the terminals in $U_i$ for all $i\leq i^*-1$.
Let $t'$ denote the common part of $t_{i^*}$ and $t_{\leq i^*-1}$.
We duplicate $t'$ and charge the cost of $t'$ to $t\cap F_{i^*-1}$.
We have
\[\cost(t')\leq \cost(t_{i^*}\cap F_{\leq i^*-1})\leq \frac{1}{1-\eps}\cdot \cost(t_{i^*}\cap F_{i^*-1})\leq \frac{1}{1-\eps}\cdot \cost(t\cap F_{i^*-1}),\]
where the second inequality follows from the definition of $\{F_i\}_i$ and using the tree structure.
%, $\cost(S')$ is at most $\frac{1}{1-\eps}$ times the cost of $S\cap F_{i^*-1}$.
We repeat on $t_{\leq i^*-1}$.
We end up with a collection of tours, each visiting only terminals in the same set $U_i$ for some $i$.
Since the charges are to disjoint parts of $t$, the overall cost of the duplicated parts is at most $\frac{1}{1-\eps}\cdot\cost(t)<(1+2\eps)\cdot \cost(t)$.
Summing over all tours in $\OPT$ concludes the proof.
\end{proof}

Using \cref{lemma:samelayer} and since $i_0\in [0,(1/\eps)-1]$ is chosen uniformly at random, we have
\begin{equation}
\label{eqn:Y-j}
\E\bigg[\sum_j \opt(Y_j)\bigg]=\eps\cdot \sum_i \opt(U_i)\leq \eps\cdot 2(1+\eps)\cdot \opt.
\end{equation}

Next, we analyze $\sum_j \opt(Z_j)$.
Consider any tour $t$ in $\OPT$.
First prune $t$ so that it does not visit terminals in the set $Y_j$ for any $j$.
The rest of the analysis is similar to the proof of Lemma~\ref{lemma:samelayer}.
Let $j^*$ be the maximum index $j$ such that $t\cap Z_j\neq \emptyset$.
Let $t_{j^*}$ denote the tour obtained by pruning $t$ so that it visits only the root and the terminals in $Z_{j^*}$.
Let $t_{\leq j^*-1}$ denote the tour obtained by pruning $t$ so that it visits only the root and the remaining terminals, i.e, the terminals in $Z_{j}$ for all $j\leq j^*-1$.
Let $t'$ denote the common part of $t_{j^*}$ and $t_{\leq j^*-1}$.
We duplicate $t'$ and we charge the cost of $t'$ to $t\cap F_{(1/\eps)\cdot j^*+i_0}$.
We have
\[\cost(t')\leq \cost(t_{j^*}\cap F_{\leq (1/\eps)\cdot j^*+i_0-1})\leq \frac{\eps}{1-\eps}\cdot \cost(t_{j^*}\cap F_{(1/\eps)\cdot j^*+i_0})\leq \frac{\eps}{1-\eps}\cdot \cost(t\cap F_{(1/\eps)\cdot j^*+i_0}),\]
where the second inequality follows from the definition of $\{F_i\}_i$ and using the tree structure.
We repeat on $t_{\leq j^*-1}$.
We end up with a collection of tours, each visiting only terminals in the same set $Z_j$ for some $j$.
Since the charges are to disjoint parts of $t$, the overall cost of the duplicated parts is at most $\frac{\eps}{1-\eps}\cdot\cost(t)<2\eps\cdot \cost(t)$.
Summing over all tours in $\OPT$, we have
\begin{equation}
\label{eqn:Z-j}
\sum_j \opt(Z_j)<(1+2\eps)\cdot \opt.
\end{equation}

From \cref{eqn:SOL,eqn:Y-j,eqn:Z-j}, we have
\[\E\big[\cost(\SOL)\big]\leq \rho\cdot (\eps\cdot 2(1+\eps) + (1+2\eps))\cdot \opt<(1+5\eps)\rho\cdot\opt.\]
This completes the proof of \cref{thm:reduction-bounded-dist}.

\section{Extension to the Splittable Tree CVRP}
In this section, we prove \cref{cor:main} by extending the PTAS in \cref{thm:main} to the splittable setting.

\label{sec:splittable}

\begin{definition}[splittable tree CVRP]
An instance of the \emph{splittable} version of the \emph{capacitated vehicle routing problem (CVRP)} on \emph{trees} consists of
\begin{itemize}
    \item an edge weighted \emph{tree} $T=(V,E)$ with $n=|V|$ and with \emph{root} $r\in V$ representing the \emph{depot},
    \item a set $V'\subseteq V$ of \emph{terminals},
    \item a positive integer \emph{demand} $d(v)$ of each terminal $v\in V'$,
    \item a positive integer \emph{tour capacity} $k$.
\end{itemize}
A feasible solution is a set of tours such that
\begin{itemize}
    \item each tour starts and ends at $r$,
    \item each tour visits at most $k$ demand,
    \item the demand of each terminal is covered, where we allow the demand of a terminal to be covered by multiple tours.
\end{itemize}
The goal is to find a feasible solution such that the total cost of the tours is minimum.
\end{definition}

%Recall that in the \emph{splittable} version of the tree CVRP, we are given a tree $T$ with $n$ terminals, and each terminal $v$ of $T$ has a positive integer \emph{demand} $d(v)$.
%The goal is to find a minimum length collection of tours starting from and ending at the depot that together cover the demand at every terminal, where the demand at a terminal may be covered by multiple tours and every tour covers at most $k$ demand.

We use a reduction from the splittable tree CVRP to the unit demand tree CVRP.
The reduction was introduced by Jayaprakash and Salavatipour~\cite{jayaprakash2021approximation}, which we summarize.
First, we reduce an instance of the splittable tree CVRP to another instance of the splittable tree CVRP in which $d(v)\leq k\cdot n$ for any terminal $v$.
Next, we replace each terminal $v$ of $T$ by a complete binary tree $T(v)$ of $d(v)$ leaves, such that each leaf of $T(v)$ is a terminal, and each edge of $T(v)$ has weight 0.
Let $\tilde T$ denote the resulting tree.
From the construction, $\tilde T$ contains at most $k\cdot n^2$ vertices.
As observed in~\cite{jayaprakash2021approximation}, the unit-demand CVRP on $\tilde T$ is equivalent to the splittable CVRP on $T$.

From \cref{thm:main}, there is an approximation scheme for the unit demand tree CVRP with running time polynomial in the number of vertices.
Therefore, we obtain an approximation scheme for the splittable tree CVRP  with running time polynomial in $n$ and $k$.

%\newpage
%\bibliographystyle{plainurl}
\bibliographystyle{alpha}
\bibliography{references}

\appendix
\newpage
\section{\cref{fig:BP19}}
\begin{figure}[h]
    \centering
    \includegraphics[scale=0.4]{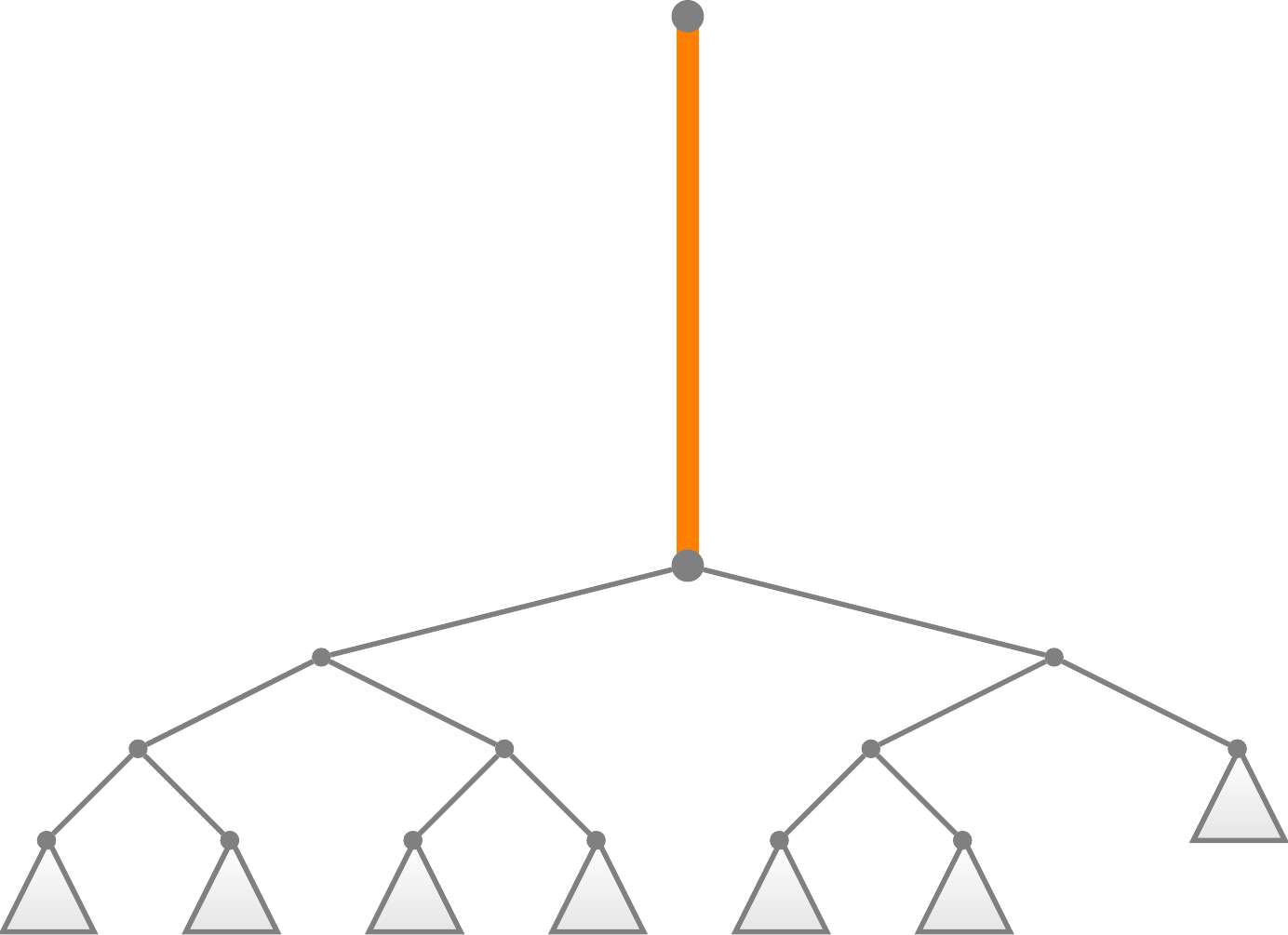}
    \caption{Bad instance for the cluster decomposition in~\cite{becker2019framework}. 
    Let $m=2\lceil\frac{1}{\delta}\rceil+1 $,  where $\delta=\Theta(\eps)$ is defined in \cite{becker2019framework}.
    Let the number of terminals in the tree be $2k$.
    The $2k$ terminals belong to $m$ subtrees of $2k/m$ terminals each (assuming that $2k$ is a multiple of $m$).
    Each of the $m$ subtrees (represented by a triangle) is a \emph{leaf cluster} according to the decomposition in \cite{becker2019framework}.
    The unique edge incident to the root vertex (thick orange edge) has weight 1, and all other edges in the tree have weight 0.
    The optimal solution consists of two tours (each visiting exactly $k$ vertices), thus has cost 4.
    When the violation of tour capacity is not allowed and when each leaf cluster is required to be visited by a single tour, the solution contains at least 3 tours, hence a cost of at least 6.
    }
    \label{fig:BP19}
\end{figure}

\section{Proof of \cref{fact:decomposition}}
\label{app:decomposition}

A \emph{leaf component} is a subtree rooted at a vertex $v\in V$ containing at least $\biggamma\cdot k$ terminals and such that each of the subtrees rooted at the children of $v$ contains strictly less than $\biggamma\cdot k$ terminals.
%Note that a leaf component is a big component, and contains at most $2\biggamma\cdot k$ terminals.
Observe that the leaf components are disjoint subtrees of $T$.
The \emph{backbone} of $T$ is the partial subtree of $T$ consisting of all edges on paths from the root of $T$ to the root of some leaf component.
\begin{definition}[key vertices]
\label{def:key-vertices}
We say that a vertex $v\in V$ is a \emph{key vertex} if $v$ is of one of the three cases: (1) the root of a leaf component; (2) a branch point of the backbone; (3) the root of the tree $T$.
\end{definition}
We say that two key vertices  $v_1,v_2$ are \emph{consecutive} if the $v_1$-to-$v_2$ path in the tree does not contain any other key vertex.
For each pair of consecutive key vertices $(v_1,v_2)$, we consider the subgraph \emph{between} $v_1$ and $v_2$, and decompose that subgraph into \emph{internal components}, each of demand at most $2\Gamma\cdot k$, such that all of these components are \emph{big} (i.e., of demand at least $\Gamma\cdot k$) except for the upmost component.
%Next, we define internal components by traversing each path of the backbone in bottom up order, starting from an \emph{exit} vertex of a new component and  gathering off-backbone subtrees as we go up the backbone, until we either just exceed $\biggamma\cdot k$ terminals or reach a branch point of the backbone: this defines an \emph{internal component} that is a partial subtree of $T$.
%Note that an internal component contains less than $2\biggamma\cdot k$ terminals.

A formal description of the  construction is given in \cref{alg:components}.
The first three properties in the claim follow from the construction.

\begin{algorithm}[t]
\caption{Decomposition into components.}
\label{alg:components}
\begin{algorithmic}[1]
\For{each vertex $v\in V$}
    \State $T(v)\gets$ subtree of $T$ rooted at $v$
    \State $n(v)\gets$ number of terminals in $T(v)$
\EndFor
\For{each non-leaf vertex $v\in V$}
    \State Let $v_1$ and $v_2$ denote the two children of $v$ in $T$
    \If{$n(v)\geq \Gamma\cdot k$ and $n(v_1)<\Gamma\cdot k$ and $n(v_2)<\Gamma\cdot k$}
        \State Let $T(v)$ be a leaf component with root vertex $v$
    \EndIf
\EndFor
\State
\State $B\gets$ set of \emph{key vertices}
\Comment {\cref{def:key-vertices}}
\For{each vertex $v_2\in B$ such that $v_2\neq r$}
    \State $v_1\gets$ lowest ancestor of $v_2$ among vertices in $B$
    \Comment{$v_1$ and $v_2$ are \emph{consecutive} key vertices}
    \For{each vertex $v$ on the $v_1$-to-$v_2$ path}
        $H(v)\gets (T(v)\setminus T(v_2))\cup \{v_2\}$
    \EndFor
    \State $x\gets v_2$
    \While{$H(v_1)$ contains at least $\Gamma\cdot k$ terminals}
        \State $v\gets$ the deepest vertex on the $v_1$-to-$x$ path such that $H(v)$ contains at least $\Gamma\cdot k$ terminals
        \State Let $H(v)$ be an internal component with root vertex $v$ and exit vertex $x$
        \State $x\gets v$
        \For{each vertex $v'$ on the $v_1$-to-$x$ path}
             $H(v')\gets (T(v')\setminus T(x))\cup \{x\}$
        \EndFor
    \EndWhile
    \If{$v_1\neq x$}
        \State Let  $H(v_1)$ be an internal component with root vertex $v_1$ and exit vertex~$x$
    \EndIf
\EndFor
\end{algorithmic}
\end{algorithm}

It remains to show the last property in the claim.
For each big component $c$, we define the image of $c$ to be itself.
It remains to consider the components that are not big, called \emph{bad} components.
Observe that the root vertex $r_c$ of any bad component $c$ is a key vertex.
We say that a bad component $c$ is a \emph{left bad} component (resp.\ \emph{right bad} component) if $c$ contains the left child (resp.\ right child) of $r_c$.
We define a map from left bad components to leaf components, such that the image of a left bad component is the leaf component that is \emph{rightmost} among its descendants, and we show that this map is injective.
Let $c_1$ and $c_2$ be any left bad components.
Observe that $r_{c_1}$ and $r_{c_2}$ are distinct key vertices.
%each branch point has at least one leaf component in its left subtree and in its right subtree.
If $r_{c_1}$ is ancestor of $r_{c_2}$ (the case when $r_{c_2}$ is ancestor of $r_{c_1}$ is similar), then the image of $c_2$ is in the left subtree of $r_{c_2}$ whereas the image of $c_1$ is outside the left subtree of $r_{c_2}$, so the images of $c_1$ and of $c_2$ are different.
In the remaining case, the subtrees rooted at $r_{c_1}$ and at $r_{c_2}$ are disjoint, so the images of $c_1$ and of $c_2$ are different.
Thus the map for left bad components is injective.
Note that every leaf component is big.
Therefore, we obtain an injective map from left bad components to big components such that the image of a left bad component is among its descendants.
The map for right bad components is symmetric.
Hence the first part of the last property. 
The second part of the last property follows from the first part of that property and the fact that the number of components with demands at least $\Gamma$ is at most $1/\Gamma$ times the total demand in the tree $T$.
This completes the proof of the claim.

%We complete the proof of \cref{fact:decomposition}.

%============================================

%====================================================

\end{document}